\documentclass[letterpaper,11pt]{article}
\usepackage{tabularx} 
\usepackage{amsmath}  
\usepackage{graphicx} 

\usepackage[margin=2.8cm,letterpaper]{geometry} 
\usepackage{xr-hyper}
\usepackage{xcolor}
\usepackage{hyperref} 
\hypersetup{
	colorlinks=true,       
	linkcolor=blue,        
	citecolor=blue,        
	filecolor=magenta,     
	urlcolor=blue         
}
\usepackage{setspace}
\usepackage{comment}

\makeatletter
\newcommand*{\addFileDependency}[1]{
  \typeout{(#1)}
  \@addtofilelist{#1}
  \IfFileExists{#1}{}{\typeout{No file #1.}}
}
\makeatother

\newcommand*{\myexternaldocument}[1]{%
    \externaldocument{#1}%
    \addFileDependency{#1.tex}%
    \addFileDependency{#1.aux}%
}
\myexternaldocument{supplementary}

\usepackage{array}
\newcolumntype{P}[1]{>{\centering\arraybackslash}p{#1}}

\newcommand{\yc}[1]{{{\color{black} #1}}}

\usepackage{authblk}

\DeclareMathOperator*{\argmin}{arg\,min}

\usepackage{amsthm}
\usepackage{amssymb}
\usepackage{multirow}
\usepackage{array,multirow}
\usepackage{booktabs}
\usepackage{subcaption}
\usepackage{mwe}
\usepackage[export]{adjustbox}
\usepackage{enumitem}
\usepackage{listings}
\usepackage{threeparttable}
\usepackage{romannum}
\usepackage{bm}
\usepackage{bbm}
\usepackage{rotating}

\usepackage[round, authoryear]{natbib}
\newtheorem{theorem}{Theorem}
\newtheorem{corollary}{Corollary}
\newtheorem{prop}{Proposition}

\usepackage[skip=1ex]{caption}
\newcommand\numberthis{\addtocounter{equation}{1}\tag{\theequation}}
\usepackage[ruled,vlined]{algorithm2e}

\usepackage{float}
\usepackage[utf8]{inputenc}
\usepackage{textcomp}
\usepackage{tabularx,ragged2e,booktabs,caption}
\newcolumntype{C}[1]{>{\Centering}m{#1}}

\usepackage{geometry}
\usepackage{float}
\title{\textbf{DIF  Statistical Inference  without Knowing Anchoring Items }}

\author[1]{Yunxiao Chen\footnote{Correspondence concerning this article should be addressed to Yunxiao Chen, Department of Statistics, London School of Economics and Political Science, London, UK
  WC2A 2AE.  E-mail: y.chen186@lse.ac.uk}}
\author[2]{Chengcheng Li}
\author[2]{Jing Ouyang}
\author[2]{Gongjun Xu}
\affil[1]{London School of Economics and Political Science}
\affil[2]{University of Michigan}

\date{}  

\begin{document}
\pagenumbering{arabic}

\maketitle

\begin{abstract}
Establishing the invariance property of an instrument (e.g., a questionnaire or test) is a key step for establishing its measurement validity. Measurement invariance is typically assessed by differential item functioning (DIF) analysis, i.e., detecting DIF items whose response distribution depends  not only on the latent trait measured by the instrument but also on the group membership. DIF analysis is confounded by the group difference in the latent trait distributions. Many DIF analyses require knowing several anchor items that are DIF-free in order to draw inferences on whether each of the rest is a DIF item, where the anchor items are used to identify the latent trait distributions. When no prior information on anchor items is available, or some anchor items are misspecified, item purification methods and regularized estimation methods can be used. The former 
iteratively purifies the anchor set by a stepwise model selection procedure, and the latter selects the DIF-free items by a LASSO-type regularization approach. Unfortunately, unlike the methods based on a correctly specified anchor set, these methods are not guaranteed to provide valid statistical inference (e.g., confidence intervals and $p$-values). In this paper, we propose a 
new method for DIF analysis under a  multiple indicators and multiple causes (MIMIC) model for DIF. 
This method adopts a minimal $L_1$ norm condition for identifying the latent trait distributions.
Without requiring prior knowledge about an anchor set, it can accurately estimate the DIF effects of individual items and further draw valid statistical inferences for quantifying the uncertainty. Specifically, the inference results allow us to control the type-I error for DIF detection, which may not be possible with item purification and regularized estimation methods. 
We conduct simulation studies to evaluate the performance of the proposed method and compare it with the anchor-set-based likelihood ratio test approach and the LASSO approach. 
The proposed method is applied to analyzing the three personality scales of the Eysenck personality questionnaire - revised (EPQ-R).  \\

\end{abstract}

\newpage 
\doublespacing
\section{Introduction}

Measurement invariance refers to the psychometric equivalence of an instrument (e.g., a questionnaire or test) across several specified groups, such as gender and ethnicity. The lack of measurement invariance suggests that the instrument has different structures or meanings to different groups, leading to biases in measurements \citep{millsap2012statistical}. 

Measurement invariance is typically assessed by differential item functioning (DIF) analysis of item response data that aims to detect the measurement non-invariant items (i.e. DIF items). More precisely, a DIF item has a response distribution that depends on not only the latent trait measured by the instrument but also respondents' group membership. Therefore, the detection of a DIF item involves comparing the item responses of different groups, conditioning on the latent traits. The complexity of the problem lies in that individuals' latent trait levels cannot be directly observed but are measured by the instrument that may contain DIF items. In addition, different groups may have different latent trait distributions. This problem thus involves 
identifying the latent trait, and then conducting the group comparison given individuals' latent trait levels. 

Many statistical methods have been developed for DIF analysis. Traditional methods for DIF analysis require prior knowledge about a set of DIF-free items, which is known as the anchor set. This anchor set is used to identify the latent trait distribution. These methods can be classified into two categories. Methods in the first category \citep{mantel1959statistical, dorans1986demonstrating,swaminathan1990detecting, shealy1993model, zwick2000using, zwick2002application, may2006multilevel, soares2009integrated, frick2015rasch} do not explicitly assume an item response theory (IRT) model, and methods in the second category  \citep{thissen1988use, lord2012applications, kim1995detection, raju1988area, raju1990determining,  woods2013langer,  oort1998simulation, steenkamp1998assessing, cao2017monte, woods2013langer, tay2015overview, tay2016item} are developed based on IRT models. Compared with non-IRT-based methods, an IRT-based method defines the DIF problem more clearly, at the price of potential model misspecification. Specifically, an IRT model represents the latent trait as a latent variable and further characterizes the item-specific DIF effects by modelling each item response distribution as a function of the latent variable and group membership.

The DIF problem is well-characterized by a multiple indicators, multiple causes (MIMIC) IRT model, which is a structural equation model  originally developed for continuous indicators \citep{zellner1970estimation,goldberger1972structural} and later extended to categorical item response data \citep{muthen1985method,muthen1991instructionally, muthen1985multiple}. A MIMIC model 
for DIF consists of a measurement component and a structural component.  The measurement component models how the item responses depend on the measured psychological trait and respondents' group membership. The structural component models the group-specific distributions of the psychological trait. The anchor set imposes zero constraints on item-specific parameters in the measurement component, making the model, including the latent trait distribution, identifiable. Consequently, the DIF effects of the rest of the items can be tested by drawing statistical inferences on the corresponding   parameters  under the identified model.

Anchor-set-based methods rely heavily on a correctly specified set of DIF-free items. The misspecification of some anchor items can lead to invalid statistical inference results \yc{-- Type I errors increase and power decreases when anchor items are not completely DIF-free}
\citep{kopf2015framework}. To address this issue, 
item purification methods \citep{candell1988iterative, clauser1993effects, fidalgo2000effects, wang2003effects, wang2004effects, wang2009mimic, kopf2015framework, kopf2015anchor} have been proposed that iteratively select an anchor set by stepwise model selection methods.
\yc{Several recently developed tree-based DIF detection methods \citep{strobl2015rasch,tutz2016item,bollmann2018item}, which can detect DIF brought by continuous covariates,
may be viewed as item purification methods.  
However, with multiple items containing DIF, item purification may suffer from masking and swamping effects \citep{barnett1994outliers}.} 
More recently, regularized estimation methods \citep{magis2015detection, tutz2015penalty, huang2018penalized,  belzak2020improving, bauer2020simplifying, schauberger2020regularization}  have   been proposed that  use   LASSO-type regularized estimation procedures for simultaneous model selection and parameter estimation.  
 \yc{Moreover, 
  \cite{bechger2015statistical} and 
  \cite{yuan2021differential} proposed DIF detection methods based on the idea of differential item
pair functioning, which does not require prior information about anchor items.}
Unfortunately, unlike many anchor-set-based methods with a correctly specified anchor set, all these methods do not provide valid statistical inference for testing the null hypothesis of ``item $j$ is DIF-free'', for each item $j$. Consequently, the type-I error for testing the hypothesis cannot be guaranteed to be controlled at a pre-specified significance level. Furthermore, although the regularised estimation methods have been shown to accurately detect DIF items, they are typically computationally intensive, since they involve solving multiple regularized maximum likelihood estimation problems with different tuning parameters.


\yc{This paper proposes a new method that addresses the aforementioned issues with the existing methods.}
The proposed method can statistically accurately and computationally efficiently estimate  the DIF effects without requiring prior knowledge about anchor items.
It draws statistical inferences on the DIF effects of individual items, yielding valid confidence intervals and p-values. The point estimation and statistical inference lead to accurate detection of the DIF items, for which \yc{the item-level type-I error and further some test-level risk (e.g., false discovery rate)} can be controlled by the inference results. 
The method is proposed under a MIMIC model with a two-parameter logistic \citep{birnbaum1968some} IRT measurement model and a linear structural model. The key to this method is a minimal $L_1$ norm assumption for identifying the true model. \yc{As will be discussed later, this assumption holds when the proportion of non-DIF items is sufficiently large.} Methods are developed for estimating the model parameters and obtaining confidence intervals and $p$-values. Procedures for the detection of DIF items are further developed.  
Our method is compared to the likelihood ratio test method  \citep{thissen1993detection} 
that requires an anchor set, and a recently proposed LASSO-based approach \citep{belzak2020improving}. 

The rest of the paper is organised as follows. In Sections~\ref{sec:setup}, we introduce a MIMIC model framework for DIF analysis. Under this model framework, a new method is proposed for the statistical inference of DIF effects in Section~\ref{sec:prop}. Related works are discussed in Section \ref{sec:related}. 
Simulation studies and a real data application are given in Sections~\ref{sec:sim} and \ref{sec:real}, respectively. We conclude with discussions in Section~\ref{sec:dis}.  All the proofs for the proposition and theorems presented in the article, and the implementation details of the proposed algorithms can be found in the Supplementary Materials.

\section{A MIMIC Formulation of DIF}\label{sec:setup}

Consider $N$ individuals answering $J$ items. Let $Y_{ij} \in \{0, 1\}$ be a binary random variable, denoting individual $i$'s response to item $j$. Let $y_{ij}$ be the observed value, i.e., the realization, of $Y_{ij}$. For the ease of exposition, we use $\mathbf Y_i = (Y_{i1}, ..., Y_{iJ})$ to denote the response vector of individual $i$. The individuals are from two groups, indicated by $x_i = 0, 1$, where 0 and 1 are referred to as the reference and focal groups, respectively. We further introduce a latent variable $\theta_i$, which represents the latent trait that the items are designed to measure. DIF occurs when the distribution of $\mathbf Y_i$ depends on not only 
$\theta_i$ but also $x_i$. More precisely, DIF occurs if $\mathbf Y_i$ is not conditionally independent of $x_i$, given $\theta_i$. Seemingly a simple group comparison problem, DIF analysis is non-trivial due to the latency of $\theta_i$. In particular, the distribution of $\theta_i$ may depend on the value of $x_i$, which confounds the DIF analysis. In what follows, we describe a MIMIC model framework for DIF analysis, under which the relationship among $\mathbf Y_i$, $\theta_i$, and $x_i$ is characterized.  It is worth pointing out that this framework can be generalized to account for more complex DIF situations; \yc{see more details in Section~\ref{sec:related}.} 

\subsection{Measurement Model}
 
The two-parameter logistic (2PL) model \citep{birnbaum1968some}  is widely used to model binary item responses (e.g., wrong/right or absent/present). In the absence of DIF, the 2PL model assumes a logistic relationship between $Y_{ij}$ and $\theta_i$, which is independent of the value of $x_i$. That is, 
\begin{equation}\label{eq:irf}
P(Y_{ij} = 1\vert \theta_i = \theta) = \frac{\exp(a_j\theta + d_j)}{1+\exp(a_j\theta + d_j)},
\end{equation}
where the slope parameter $a_j$ and intercept parameter $d_j$ are typically known as the discrimination and easiness parameters, respectively. The right hand side of (\ref{eq:irf}) as a function of $\theta$ is known as the 2PL item response function. 
When the items potentially suffer from DIF, the item response functions may depend on the group membership $x_i$. In that case, the item response function can be modeled as 
\begin{equation}\label{eq:irf_dif}
P(Y_{ij} = 1\vert \theta_i = \theta, x_i) = \frac{\exp(a_j\theta + d_j + \gamma_j x_i)}{1+\exp(a_j\theta + d_j + \gamma_j x_i)},
\end{equation}
where $\gamma_j$ is an item-specific parameter that characterizes its DIF effect. More 
precisely, 
$$ \frac{P(Y_{ij} = 1\vert \theta_i = \theta, x_i = 1)/(1- P(Y_{ij} = 1\vert \theta_i = \theta, x_i = 1))}{P(Y_{ij} = 1\vert \theta_i = \theta, x_i=0)/(1-P(Y_{ij} = 1\vert \theta_i = \theta, x_i=0))} = \exp(\gamma_j).$$
That is, $\exp(\gamma_j)$ is the odds ratio for comparing two individuals from two groups who have the same latent trait level. Item $j$ is DIF-free under this model when $\gamma_j = 0$. We further make the local independence assumption as in most IRT models. That is, $Y_{i1}$, ..., $Y_{iJ}$ are assumed to be conditionally independent, given $\theta_i$ and $x_i$.

 
 
\subsection{Structural Model} \label{subsec:str}

A structural model specifies the distribution of $\theta_i$, which may depend on the group membership. Specifically, we assume the conditional distribution of $\theta_i$ given $x_i$ to follow a normal distribution, 
$$\theta_i \vert x_i \sim N(\beta x_i, 1_{\{x_i = 0\}} + \sigma^2 1_{\{x_i = 1\}}).$$ 
Note that the latent trait distribution for the reference group is set to a standard normal distribution to identify the location and scale of the latent trait. A similar assumption is typically adopted in IRT models for a single group of individuals. 

The MIMIC model for DIF combines the above measurement and structural models, for which a path diagram is given in Figure~\ref{fig:path}. The marginal likelihood function for this MIMIC model takes the form 
\begin{equation}\label{eq:likelihood}
L(\Xi) = \prod_{i=1}^N \int \left(\prod_{j=1}^J \frac{\exp(y_{ij}(a_j\theta + d_j + \gamma_j x_i))}{1+\exp(a_j\theta + d_j + \gamma_j x_i)}\right)\frac{1}{\sqrt{2\pi}} \exp\left(\frac{-(\theta - \beta x_i)^2}{2(1_{\{x_i = 0\}} + \sigma^2 1_{\{x_i = 1\}})}\right)  d\theta,
\end{equation}
where $\Xi = \{\beta, \sigma^2, a_j, d_j, \gamma_j, j = 1, ..., J\}$ denotes all the fixed model parameters. 

The goal of DIF analysis is to detect the DIF items, i.e., the items for which $\gamma_j \neq 0$. Unfortunately, without further assumptions, this problem is ill-posed due to the non-identifiability of the model. We discuss this identifiability issue below. 


\begin{figure}
    \centering
    \includegraphics[scale = 2]{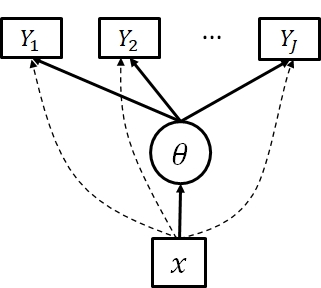}
    \caption{The path diagram of a MIMIC model for DIF analysis. The subscript $i$ is omitted for simplicity. The dashed lines from $x$ to $Y_j$ indicate the DIF effects. }
    \label{fig:path}
\end{figure}

\subsection{Model Identifiability}\label{subsec:identi}

Without further assumptions, the above MIMIC model is not identifiable. 
That is, for any constant $c$, the model remains equivalent, if we simultaneously replace $\beta$ and $\gamma_j$ by $\beta+c$ and $\gamma_j -a_jc$, respectively, and keep $a_j$ unchanged. This identifiability issue is due to that all the items are allowed to suffer from DIF, resulting in an unidentified latent trait. In other words, without further assumptions,  it is impossible to disentangle the DIF effects and the difference between the latent trait distributions of the two groups.   

\yc{According to Theorem 8.3 of \cite{san2015identification}, the location shift described above is the only source of indeterminacy for this MIMIC model when $J\geq 3$ and the sizes of both groups go to infinity. Let $\Xi^* = \{\beta^*, (\sigma^*)^2, a_j^*, d_j^*, \gamma_j^*, j = 1, ..., J\}$ be a set of parameters for the true model. Then a set of parameters yields the same data distribution as the true model if and only if there exists a constant $c$ such that $\Xi^*(c) = \{\beta^*+c, (\sigma^*)^2, a_j^*, d_j^*, \gamma_j^* - a_j^*c, j = 1, ..., J\}$. That is, $\{\Xi^*(c): c \in \mathbb R\}$ gives an equivalent class for the true model parameters. Knowing one or more anchor items means that the corresponding $\gamma_j^*$s are known to be zero, which fixes the location indeterminacy. However, if no anchor item is known,   we need to answer the   question:  
which member of this equivalent class should be used to define DIF effects?   We address it in Section~\ref{sec:prop} below.
}


\section{Proposed Method}\label{sec:prop}

\yc{In what follows, we address the model identifiability problem raised above and then 
propose a new method for DIF analysis that does not require prior knowledge about anchor items. 
As will be shown in the rest, the proposed method can not only accurately detect the DIF items, but  also provide valid statistical inference for testing the hypotheses of $\gamma_j = 0$, for any $j = 1, ..., J$.}


\subsection{\yc{Model Identifiability, Sparsity, and Minimal $L_1$ Condition}}\label{subsec:ml1}

\yc{We now address the model identifiability problem. The most natural idea is to choose $\Xi^*$ as the true parameter vector when the corresponding  $\boldsymbol\gamma^* = (\gamma_1^*, ..., \gamma_J^*)^\top$ is the sparsest in the equivalent class $\{\Xi^*(c): c \in \mathbb R\}$. In other words, we say  $\Xi^*$ is the true model parameter when 
$\Vert\boldsymbol\gamma ^*\Vert_0 < \Vert\boldsymbol\gamma ^*(c)\Vert_0$ for any $c$, where 
$\boldsymbol\gamma^*(c) = (\gamma_1^* - a_1^*c, ..., \gamma_J^* - a_J^*c)^\top$ and $\Vert\cdot\Vert_0$ denotes the $L_0$ norm, i.e., the number of non-zero entries in a vector. We note that this sparsity assumption is essential if one wants to formulate the DIF detection problem as a model selection problem.  
It is explicitly or implicitly made by most DIF detection methods that do not require anchor items, including item purification and regularised estimation methods.  

However, there are still ambiguities with 
the above sparsity assumption.  First,  how sparse should $\boldsymbol\gamma^*$ be? We note that the true parameters need to satisfy $\Vert\boldsymbol\gamma^*\Vert_0 \leq J-2$; that is, there are at least two zeros in the true $\boldsymbol\gamma^*$ vector. This is because, $\Vert\boldsymbol\gamma^*(c)\Vert_0 \leq J-1$, if $c = \gamma_j^*/a_j^*$ for any $j$. However, we may not want to identify the latent trait with only two items because, in that case, the two items are of high leverage -- the latent trait becomes unidentifiable when we remove one of these items. For the latent trait to be firmly identified, it may be sensible to assume that $\boldsymbol\gamma^*$ is sufficiently sparse, where the sparsity may be measured by  the proportion of non-zero coefficients in $\boldsymbol\gamma^*$. Further discussions will be provided in the sequel regarding the sparsity level. We note that this ``sufficiently sparse" assumption aligns well with the practical utility of DIF analysis in educational testing \citep[e.g.,][]{holland1993differential} as well as certain settings of psychological measurement \citep[e.g., Chapter 1,][]{millsap2012statistical} and health-related measurement \citep[e.g.,][]{scott2010differential}. For example, in educational testing, DIF analysis is conducted to ensure the fairness of a test form. In this application, the test operator aims to identify a small number of DIF items that cause a bias in the test result.
The identified items will be reviewed by domain experts, and then revised or removed from the item pool. For this process to be operationally feasible, one typically needs to assume that the majority of the items are DIF-free, i.e., $\boldsymbol\gamma^*$ is sufficiently sparse. 

Second, the $L_0$ norm is not easy to work with  from a statistical perspective. Due to the randomness in the data, likelihood-based estimation methods almost never give us a truly sparse solution. Consequently, one essentially needs to search over $O(2^J)$ all possible models to find the sparest model (e.g., using a suitable information criterion). Item purification and regularized estimation methods narrow the search by stepwise procedures and regularized estimation procedures, respectively. Even with these methods,   the computation can still be intensive, and consistent selection of the true model is not always guaranteed. 
}

\yc{Following the previous discussions, we now impose a condition for identifying the true model parameters, which is statistically easy to work with and suitable when  the true $\boldsymbol\gamma^*$ is sufficiently sparse. Specifically,}
we require the following minimal $L_1$ (ML1) condition to hold 
\begin{equation}\label{eq:ML1}
\sum_{j=1}^J \vert \gamma_j^* \vert < \sum_{j=1}^J \vert \gamma_j^*  -a_j^*c\vert,
\end{equation}
 for all $c \neq 0$. This assumption implies that, among all models that are equivalent to the true model, the true parameter vector $\boldsymbol{\gamma}^*$ has the smallest $L_1$ norm. 
 Equivalently, we can rewrite \eqref{eq:ML1} as 
 \begin{equation}\label{eq:ML2}
 \arg\min_c h(c) = 0,
 \end{equation}
 where  $h(c) = \sum_{j=1}^J \vert \gamma_j^*  -a_j^*c\vert$. We give an example of 
 $h(c)$ in Figure~\ref{fig:ML1}, 
where $h(c)$ is constructed with a sparse $\boldsymbol\gamma^*$. More specifically, we construct $h(c)$ with $J = 10$, $a_j^* = 1$ for all $j$, $\gamma_j^* = 0$ and 1 when $j = 1, ..., 8$ and  $j = 9, 10$, respectively. \yc{In this example, we note
that $h(c)$ has a unique minumum at $c=0$, i.e., 
\eqref{eq:ML1}, or equivalently, \eqref{eq:ML2} holds.}

 \yc{In what follows, we show that the ML1 condition holds when $\boldsymbol\gamma^*$ is sufficiently sparse.} The following proposition provides a sufficient and necessary condition for the ML1 condition \eqref{eq:ML1} (or equivalently \eqref{eq:ML2}) to hold.  The proof is given in the Supplementary Materials.

\begin{figure}
    \centering
    \includegraphics[scale = 0.6]{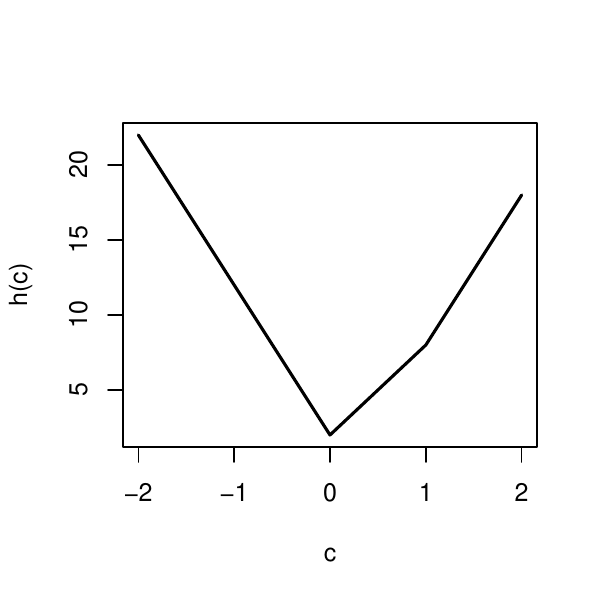}
    \caption{Function $h(c) = \sum_{j=1}^J \vert \gamma_j^*  -a_j^*c\vert$, where $J = 10$, $a_j^* = 1$ for all $j$, $\gamma_j^* = 0$ and 1 for $j = 1, ..., 8$ and  $j = 9, 10$, respectively. The minimal value of $h(c)$ is achieved when $c = 0$. }
    \label{fig:ML1}
\end{figure}
 
 \begin{prop}\label{prop:ML1} 
 \yc{Assume that $a_j^* \neq 0$ for all $j$.}
 Condition \eqref{eq:ML1} holds if and only if 
\begin{equation}\label{eq:1prop1}
 \sum_{j=1}^J |a_j^*| \left(- I\big(\frac{\gamma_j^*}{a_j^*} \geq 0\big) + I\big(\frac{\gamma_j^*}{a_j^*} < 0\big)\right) < 0
\end{equation}
 and 
  \begin{equation}\label{eq:2prop1}
  \sum_{j=1}^J |a_j^*| \left(- I\big(\frac{\gamma_j^*}{a_j^*} > 0\big) + I\big(\frac{\gamma_j^*}{a_j^*} \leq 0\big)\right) > 0,
  \end{equation}
  where $I(\cdot)$ is the indicator function. 
 \end{prop}
 
We note that inequalities  \eqref{eq:1prop1} and \eqref{eq:2prop1}
hold for the example in Figure~\ref{fig:ML1}, where $\sum_{j=1}^J I(\gamma_j^* \geq 0) = 10$, $\sum_{j=1}^J  I(\gamma_j^* < 0) =0$, $\sum_{j=1}^J I(\gamma_j^* \leq 0) = 8$ and  $\sum_{j=1}^J  I(\gamma_j^* > 0) = 2$. \yc{To elaborate on the results of Proposition~\ref{prop:ML1}, we first consider a special case when $a_j^* = 1$ for all $j$, i.e., the measurement model is a one-parameter logistic model when there is no DIF.} Then according to Proposition~\ref{prop:ML1}, the ML1 condition holds if and only if 
$\sum_{j=1}^J I(\gamma_j^* \geq 0) > \sum_{j=1}^J  I(\gamma_j^* < 0) $
and 
$\sum_{j=1}^J I(\gamma_j^* \leq 0) > \sum_{j=1}^J  I(\gamma_j^* > 0).$ \yc{Suppose that more than half of the items are DIF-free, i.e., $\sum_{j=1}^J I(\gamma_j^* = 0) > J/2$. Then the ML1 condition holds, because 
$\sum_{j=1}^J I(\gamma_j^* \geq 0) > J/2 > \sum_{j=1}^J  I(\gamma_j^* < 0) $
and 
$\sum_{j=1}^J I(\gamma_j^* \leq 0) > J/2 > \sum_{j=1}^J  I(\gamma_j^* > 0).$  
More generally, let $\gamma_{(1)}^* \leq \gamma_{(2)}^* \leq \cdots \leq  \gamma_{(J)}^*$ be the order statistics of $\gamma_1^*$, ..., $\gamma_J^*$. 
The ML1 condition holds when $\gamma_{((J+1)/2)}^* = 0$ if $J$
is an odd number, and when $\gamma_{(J/2)}^* = \gamma_{((J/2)+1)}^* = 0$ if $J$ is an even number. That is, the ML1 condition holds when we have similar numbers of positive and negative DIF items and a few non-DIF items, in which case  the ML1 condition can hold even if $\sum_{j=1}^J I(\gamma_j^* = 0) \leq J/2$. However, if all the DIF items are of the same direction (all positive or all negative), then it is easy to show that the ML1 condition does not hold if $\sum_{j=1}^J I(\gamma_j^* = 0) \leq J/2$.

We then extend the above discussion to the general setting where the discrimination parameters vary across items. Based on Proposition~\ref{prop:ML1}, we provide a sufficient condition for the ML1 condition, which suggests that the ML1 condition holds when $\gamma_j^* = 0$ for a sufficient number of items.  
}
 
\begin{corollary}\label{coro:1}
\yc{Assume that $a_j^* \neq 0$ for all $j$. Let $\rho^* = \max_j\{\vert a_j^*\vert\}/\min_j\{\vert a_j^*\vert\}$. Then 
Condition \eqref{eq:ML1} holds if 
\begin{equation}\label{eq:suff1}
\sum_{j=1}^J I(\gamma_j^*/a_j^* \leq  0) > \rho^* \sum_{j=1}^J I(\gamma_j^*/a_j^* >  0)
\end{equation}
and 
\begin{equation}\label{eq:suff2}
\sum_{j=1}^J I(\gamma_j^*/a_j^* <  0) >  \rho^* \sum_{j=1}^J I(\gamma_j^*/a_j^* \geq  0).
\end{equation}
}

\end{corollary}

\yc{We note \eqref{eq:suff1} and \eqref{eq:suff2} are not a necessary condition, meaning that the ML1 condition can still hold even if \eqref{eq:suff1} and \eqref{eq:suff2} are not satisfied. Here, $\rho^*$ quantifies the variation of the absolute discrimination parameters, where a larger value of $\rho^*$ indicates a higher variation. Corollary~\ref{coro:1} suggests that ML1 condition holds if $\sum_{j=1}^J I(\gamma_j^* = 0) > (\rho^*/(1+\rho^*))J$. For instance, when $\rho^* = 2$, then the ML1 condition is guaranteed  if $\sum_{j=1}^J I(\gamma_j^* = 0) >  (2/3)J$, i.e., at least two-thirds of the items are DIF-free. This sparsity requirement can be relaxed if the sizes of items with 
$\gamma_j^*/a_j^* >  0$ and those with $\gamma_j^*/a_j^* < 0$
are balanced. 
}

\subsection{Parameter Estimation}

\yc{Suppose that the true model parameters satisfy the ML1 condition. Then these parameters can be estimated by finding the ML1 estimate $\hat \Xi = \{\hat \beta, \hat \sigma^2, \hat \gamma_j, \hat d_j, \hat a_j, j=1,...,J\} $ satisfying 
\begin{equation}
\log L(\hat \Xi) = \max_{\Xi} \log L(\Xi)
\end{equation}
and for any $\tilde \Xi = \{\tilde \beta, \tilde \sigma^2, \tilde \gamma_j, \tilde d_j, \tilde a_j, j=1,...,J\} $ satisfying $\log L(\tilde \Xi) = \max_{\Xi} \log L(\Xi)$, 
\begin{equation}
\sum_{j=1}^J \vert \hat \gamma_j \vert \leq  \sum_{j=1}^J \vert \tilde \gamma_j\vert.
\end{equation}
That is, $\hat \Xi$ is a maximum likelihood estimate whose DIF parameter vector has the smallest $L_1$ norm.  The estimate $\hat \Xi$ can be easily computed using a two-step procedure as described in Algorithm~\ref{alg:ML1}. }


\begin{algorithm}[h!]
\SetAlgoLined
\begin{itemize}
\item[] {\bf Step 1:} Solve the following MML estimation problem 

\begin{align}
\tilde \Xi &= \arg\max_{\Xi} \log L(\Xi),
\quad s.t. \quad\gamma_1 = 0. \label{eq:mml}
\end{align}

\item[] {\bf Step 2:} Solve the optimization problem 
\begin{equation}\label{eq:lad}
\hat c = \arg\min_{c} \sum_{j=1}^J \vert \tilde\gamma_j  - \tilde a_jc\vert
\end{equation}
\item[] {\bf Output:} The ML1 estimate  
$\hat \gamma_j = \tilde\gamma_j  - \tilde a_j\hat c,$
$\hat\beta = \tilde\beta + \hat c,$   $\hat \alpha_j = \tilde \alpha_j$, $\hat d_j =\tilde d_j , {\hat{\sigma}^2 = \tilde{\sigma}^2}$.
\end{itemize}
\caption{}
\label{alg:ML1}
\end{algorithm}

We provide some remarks about these steps.
The estimator (\ref{eq:mml}) in Step 1 can be viewed as the MML estimator of the MIMIC model, treating item 1 as an anchor item. 
We emphasize that the constraint $\gamma_1 = 0$ in Step 1 is an arbitrary but mathematically convenient constraint for ensuring the estimability of the MIMIC model when solving  (\ref{eq:mml}). 
It does not require item 1 to be truly a DIF-free item. This constraint can be replaced by any equivalent constraint, for example, $\gamma_2 = 0$, while not affecting the final estimation result.
Step 2 finds the transformation that leads to the ML1 solution among all the models equivalent to the estimated model from Step 1.  The optimization problem (\ref{eq:lad}) is convex that takes the same form as the Least Absolute Deviations (LAD) objective function in median regression  \citep{koenker_2005}. {Specifically, the LAD function is a statistical optimization function measuring the sum of absolute residuals. Given a set of data $(x_i, y_i)$ for $i = 1,\dots, n$, the LAD  function is defined as $S(f) = \sum_{i=1}^{n} \vert y_i - f(x_i) \vert$ and we seek to find $f$ that minimizes LAD function $S$. Our problem (\ref{eq:lad}) is convex since we are minimizing a convex LAD function over a set of real numbers, which gives us a unique global optimum.}  Consequently, it can be solved using standard statistical packages/software for quantile regression. The R package ``\textit{quantreg}''  \citep{quantreg} is used in our simulation study and real data analysis. 

The ML1 condition~(\ref{eq:ML1}), together with some additional regularity conditions, guarantees the consistency of the above ML1 estimator. That is, $\hat \Xi$ will converge to $\Xi^*$ as the sample size $N$ grows to infinity. This result is formalized in Theorem~\ref{thm:consistency}, with its proof given in the Supplementary Materials.

\begin{theorem}\label{thm:consistency}
Let  $\Xi^*=\{\beta^*, (\sigma^*)^2, \gamma_j^*, d_j^*, a_j^*, j=1,..., J\}$ be the true model parameters, and $\Xi^\dagger=\{\beta^\dagger, (\sigma^2)^\dagger, \gamma_j^\dagger, d_j^\dagger, a_j^\dagger, j=1,...,J\}$ be the true parameter values of the equivalent MIMIC model with constraint $\gamma_1^\dagger=0$. Assume this equivalent model satisfies the standard regularity conditions in Theorem 5.14 of \cite{van2000asymptotic} that concerns the consistency of maximum likelihood estimation. Further, assume that the ML1 condition~(\ref{eq:ML1}) holds. 
Then $\vert \hat \beta - \beta^*\vert=o_P(1), {\vert  \hat{\sigma}^2 -(\sigma^2)^*\vert =o_P(1)}, \vert \hat \gamma_j - \gamma_j^*\vert =o_P(1)$, $\vert \hat a_j - a_j^*\vert =o_P(1), \mbox{and~} \vert \hat d_j - d_j^*\vert =o_P(1)$ as $N\to \infty$.
\end{theorem}

\yc{With a consistent point estimator, one can consistently 
select the true model, i.e., identifying the zeros and non-zeros in $\boldsymbol\gamma^*$, using a hard-thresholding procedure~\citep[see e.g.][]{meinshausen2009lasso}. As our focus is on the statistical inference of DIF parameters, we skip the details of the hard-thresholding procedure here.}

\subsection{Statistical Inference}
The statistical inference of individual $\gamma_j$ parameters  is of particular interest in the DIF analysis. With the proposed estimator, we can draw valid statistical inference on  the DIF parameters $\gamma_j$. 

Note that the uncertainty of $\hat \gamma_j$ is inherited from $\tilde \Xi$, where 
 $\sqrt{N}(\tilde \Xi -\Xi^\dagger)$  asymptotically follows a mean-zero multivariate normal distribution\footnote{Note that this is a degenerated multivariate normal distribution since $\tilde \gamma_1 = \gamma_1^\dagger=0$.}
 by the large-sample theory for maximum likelihood estimation; see Supplementary Materials for more details.
 We denote this multivariate normal distribution by $N(\mathbf 0, \Sigma^*)$, where a consistent estimator of $\Sigma^*$, denoted by $\hat \Sigma_N$, can be obtained based on the marginal likelihood.
 We define a function 
 $$G_j(\Xi) = \gamma_j - a_j \times \arg\min_{c} \sum_{l=1}^J \vert \gamma_l  -  a_lc\vert,$$
 where $\Xi = \{\beta, \sigma^2, a_l, d_l, \gamma_l, l = 1, ..., J\}$. Note that the function $G_j$ maps an arbitrary parameter vector of the MIMC model to  the $\gamma_j$ parameter of the equivalent ML1 parameter vector.
To draw statistical inference, we need the distribution of $$\hat \gamma_j - \gamma_j^* = G_j(\tilde \Xi) - G_j(\Xi^\dagger).$$
By the asymptotic distribution of $\sqrt{N}(\tilde \Xi -\Xi^\dagger)$, we know that the distribution of $G_j(\tilde \Xi) - G_j(\Xi^\dagger)$ can be approximated by that of $G_j(\Xi^\dagger + \mathbf Z/\sqrt{N}) - G_j(\Xi^\dagger)$, and the latter can be further approximated by 
$G_j(\tilde \Xi + \mathbf Z/\sqrt{N}) - G_j(\tilde\Xi)$, where $\mathbf Z$ follows a normal distribution 
$N(\mathbf 0, \hat \Sigma_N)$. Therefore, although function $G_j$ does not have an analytic form, we can approximate the distribution of $\hat \gamma_j - \gamma_j^*$ by Monte Carlo simulation. We summarize this procedure in Algorithm~\ref{alg:bootstrap} below. 

\begin{algorithm}[h!]
\begin{itemize}
\item[] {\bf Input:} The number of Monte Carlo samples $M$ and significance level $\alpha$. 
\item[] {\bf Step 1:} Generate $M$ i.i.d. samples from a multivariate normal distribution with mean $\mathbf 0$ and covariance matrix $\hat \Sigma_N$. We denote these samples as $\mathbf Z_1$, ..., $\mathbf Z_M$. 

\item[] {\bf Step 2:} Obtain $e_m = G_j(\tilde \Xi + \mathbf Z_m/\sqrt{N}) - G_j(\tilde\Xi)$, for $m=1, ..., M$.

\item[] {\bf Step 3:} Obtain the $\alpha/2$ and $1-\alpha/2$ quantiles of the empirical distribution of $(e_1, ..., e_M)$, denoted by $q_{\alpha/2}$ and $q_{1-\alpha/2}$, respectively. 

\item[] {\bf Output:} Level $1-\alpha$ confidence interval for $\gamma_j^*$ is given by $(\hat \gamma_j - q_{1-\alpha/2}, \hat \gamma_j - q_{\alpha/2})$. In addition, the $p$-value for a two-sided test of $\gamma_j^* = 0$ is given by 
$$\frac{\sum_{i=1}^M 1_{\{|e_i| > |\hat \gamma_j|\}} }{M}.$$
\end{itemize}

\caption{}
\label{alg:bootstrap}
\end{algorithm}

Algorithm~\ref{alg:bootstrap} only involves sampling from a multivariate normal distribution and solving a convex optimization problem based on the LAD objective function, both of which are computationally efficient. The value of $M$ is set to 10,000 in our simulation study and 50,000 in the real data example below.

The p-values can be used to control the type-I error rate, i.e.,  
the probability of mistakenly detecting a non-DIF item as a DIF one. To control the item-specific type-I errors to be below a pre-specified threshold $\alpha$ (e.g., $\alpha = 0.05$), we detect the items for which the corresponding p-values   are below $\alpha$. 
Besides the type-I error, we may also consider the False Discovery Rate (FDR) for DIF detection \citep{bauer2020simplifying}
to account for multiple comparisons, where the FDR is defined as the expected ratio of the number of non-DIF items to the total number of detections.  To control the FDR, the Benjamini-Hochberg (B-H) procedure \citep{benjamini1995controlling} can be employed to the p-values. \yc{Other compound risks may also be considered, such as the familywise error rate.}

\section{Related Works \yc{and Extensions}}\label{sec:related} 

\subsection{\yc{Related Works}}

Many of the IRT-based DIF analyses  \citep{thissen1986beyond,thissen1988use, thissen1993detection} 
require prior knowledge about a subset of DIF-free items, which are known as the anchor items. More precisely, we denote this known subset by $A$. Under the MIMIC model described above, it implies that the constraints $\gamma_j = 0$ are imposed for all $j \in A$ in the estimation.
With these zero constraints, the $\gamma_j$ parameters cannot be freely transformed, and thus, the above MIMIC model becomes identifiable. Therefore, for each non-anchor item $j \notin A$, the 
hypothesis of $\gamma_j = 0$ can be tested, for example, by a likelihood ratio test. 
The DIF items can then be detected based on the statistical inference of these hypothesis tests.

The validity of the anchor-item-based analyses relies on the assumption that the anchor items are truly DIF-free. If the anchor set includes one or more DIF items, then the results can be misleading \citep{kopf2015framework}. To address the issue brought by the mis-specification of the anchor set, item purification methods \citep{candell1988iterative, clauser1993effects, fidalgo2000effects, wang2003effects, wang2004effects, wang2009mimic, kopf2015framework, kopf2015anchor} have been proposed that iteratively purify the anchor set. These methods conduct model selection using a stepwise procedure to select the anchor set, implicitly assuming that there exists a reasonably large set of DIF items. Then DIF is assessed by hypothesis testing given the selected anchor set.  
This type of methods also has several limitations. First, 
the model selection results may be sensitive to the choice of the initial set of anchor items and the specific stepwise procedure (e.g., forward or backward selection), which is a common issue with stepwise model selection procedures (e.g., stepwise variable selection for linear regression). Second, the model selection step has uncertainty. As a result, there is no guarantee that the selected anchor set is completely DIF-free, and furthermore, the post-selection statistical inference of items may not be valid in the sense that the type-I error may not be controlled at the targeted significance level.

\yc{\cite{bechger2015statistical} and \cite{yuan2021differential} proposed DIF detection methods based on the idea of differential item
pair functioning. They considered a one-parameter logistic model setting, which corresponds to the case when $a_1 = \cdots = a_J$  in the current MIMIC model. Their idea is that the difference $\gamma_j - \gamma_{j'}$ is identifiable for any $j\neq j'$, though each individual $\gamma_j$ is not identifiable due to location indeterminacy. \cite{bechger2015statistical} focused on testing $\gamma_j - \gamma_{j'} = 0$ for all item pairs, and 
\cite{yuan2021differential} proposed data visualization methods and a Monte Carlo test to identify individual DIF items. However, they did not provide statistical inferences for the DIF effects of individual items. 

 }

Regularized estimation methods \citep{magis2015detection, tutz2015penalty, huang2018penalized,  belzak2020improving, bauer2020simplifying, schauberger2020regularization}  have also been proposed for identifying the anchor items, which also implicitly assumes that many items are DIF-free. These methods 
do not require prior knowledge about anchor items, and simultaneously select the DIF-free items and estimate the model parameters using a LASSO-type  penalty \citep{tibshirani1996regression}. Under the above MIMIC model, a regularized estimation procedure  solves the following optimization problem, 
\begin{equation}\label{eq:LASSO}
\hat \Xi^\lambda = \arg\min_{\Xi} - \log L(\Xi) + \lambda \sum_{j=1}^J \vert \gamma_j\vert,
\end{equation}
where $\lambda > 0$ is a tuning parameter that determines the sparsity level of the estimated $\gamma_j$ parameters. Generally speaking, a larger  value of $\lambda$ leads to a more sparse vector $\hat{\boldsymbol{\gamma}}^\lambda = ({\hat \gamma}_1^\lambda, ..., \hat \gamma_J^\lambda).$ A regularization method \citep[e.g.][]{belzak2020improving} solves the optimization problem (\ref{eq:LASSO}) for a sequence of $\lambda$ values, and then selects the tuning parameter $\lambda$ based on the BIC. Let $\hat \lambda$ be the selected tuning parameter. Items for which ${\hat \gamma}_j^{\hat\lambda} \neq 0$ are  
classified as DIF items and the rest are classified as  DIF-free items. While the regularization methods are computationally more stable than  stepwise model selection in the item purification methods, they also suffer from some limitations. First, they involve solving non-smooth optimization problems like (\ref{eq:LASSO}) for different tuning parameter values, which is not only computationally intensive but also requires tailored computation code that is not available in most statistical packages/software for DIF analysis. Second, these methods may be sensitive to the choice of the tuning parameter. Although methods and theories have been developed in the statistics literature to guide the selection of the tuning parameter, there is no consensus on how the tuning parameter should be chosen, leaving ambiguity in the application of these methods. 
Third, from the theoretical perspective, it is not clear whether these methods can guarantee model selection consistency. In particular, the model selection consistency of the LASSO procedure almost always requires a strong assumption called the 
irrepresentable condition \citep{zhao2006model,van2009conditions}. 
It is not clear when this assumption holds for the current problem.
On the other hand, the proposed ML1 condition is much easier to understand and check, as discussed in Section~\ref{subsec:ml1}. 
Finally, as a common issue of regularized estimation methods, obtaining valid statistical inference from these methods is not straightforward. 
That is, regularized estimation like (\ref{eq:LASSO}) does not provide a valid $p$-value for testing the null hypothesis $\gamma_j = 0$. In fact, post-selection inference after regularized estimation was conducted in \cite{bauer2020simplifying}, where the type I error cannot be controlled at the targeted level under some simulation scenarios.

We notice that there is a connection between the proposed estimator and the regularized estimator (\ref{eq:LASSO}). Note that $\hat \Xi$ is the one with the smallest $\sum_{j=1}^J \vert \gamma_j\vert$ among all equivalent estimators that maximize the likelihood function~(\ref{eq:likelihood}). When the solution path of (\ref{eq:LASSO}) is smooth and the solution to the ML1 problem (\ref{eq:lad}) is unique,  it is easy to see that $\hat \Xi$ is the limit of $\hat \Xi^{\lambda}$ when the positive tuning parameter $\lambda$ converges to zero. In other words, the proposed estimator can be viewed as a limiting version of the LASSO estimator (\ref{eq:LASSO}). According to Theorem~\ref{thm:consistency}, this limiting version of the LASSO estimator is sufficient for the consistent estimation of the true model under the ML1 condition.

\yc{We clarify that the proposed method may not always outperform other methods in terms of accuracy in classifying items, such as the LASSO procedure. From the simulation results in Section~\ref{sec:sim} below, we see that the proposed method and the LASSO procedure have similar accuracy in item classification when the DIF parameters are large. The key advantage of the proposed method  is that the proposed one provides valid statistical inference (e.g., p-values) when anchor items are not available. The inference results allow us to tackle the uncertainty in the decisions of DIF detection, which can be useful in many applications of DIF analysis where high-stake decisions need to be made.  }


\subsection{\yc{Extensions}}

\yc{
While we focus on the two-group setting and uniform-DIF (i.e., only the intercepts depend on the groups) in the previous discussion, the proposed framework is very general that can be easily generalised to other settings. In what follows, we discuss the ML1 condition under different settings. The proposed methods for point estimation and statistical inference can be extended accordingly.

\paragraph{Non-uniform DIF.} Under the 2PL measurement model, non-uniform DIF happens when the discrimination parameter also differs across groups. To model non-uniform DIF, we extend the current measurement model \eqref{eq:irf_dif} to 
\begin{equation}\label{eq:irf_dif2}
P(Y_{ij} = 1\vert \theta_i = \theta, x_i) = \frac{\exp(a_j\exp(\zeta_j x_i)\theta + d_j + \gamma_j x_i)}{1+\exp(a_j\exp(\zeta_j x_i)\theta + d_j + \gamma_j x_i)},
\end{equation}
while keeping the structural model the same as in Section~\ref{subsec:str}. 
This extended model has both location and scale indeterminacies.  Let $\Xi^* = \{\beta^*, (\sigma^*)^2, a_j^*, d_j^*, \zeta_j^*, \gamma_j^*, j = 1, ..., J\}$ be a set of parameters for the true model. Then a set of parameters yields the same data distribution as the true model if there exist constants $m$ and $c$ such that $\Xi^*(m,c) = \{(\beta^*-c)\times\exp(-m), \exp(-2m)\times(\sigma^*)^2, a_j^*, d_j^*, \zeta_j^* + m, \gamma_j^* - ca_j^*\exp(\zeta_j^*), j = 1, ..., J\}$. Note that an item $j$ is DIF-free if $\zeta_j = \gamma_j =0$. Under the same spirit as the ML1 condition~\eqref{eq:ML1}, we may assume the true model parameters $\Xi^*$ to satisfy 
$$\sum_{j=1}^J \vert a_j^* \vert < \sum_{j=1}^J \vert a_j^* +m \vert, \mbox{~and~} \sum_{j=1}^J\vert\gamma_j^*\vert < \sum_{j=1}^J\vert\gamma_j^* + ca_j^*\exp(\zeta_j^*)\vert$$ when $m\neq 0$ and $c\neq 0$.
These conditions tend to be satisfied when the proportion of DIF-free items is sufficiently large.

\paragraph{Multi-group setting.} There may be more than two groups in some DIF applications. Suppose that there are $K+1$ groups -- one reference group and $K$ focal groups. Let $x_i \in \{0, ..., K\}$
indicate the group membership. 

For simplicity, we focus on the 
uniform DIF setting. Then the measurement model becomes 
\begin{equation}\label{eq:irf_dif3}
P(Y_{ij} = 1\vert \theta_i = \theta, x_i =k) = \frac{\exp(a_j \theta + d_j + \gamma_{jk} )}{1+\exp(a_j \theta + d_j + \gamma_{jk})}, k =  1,..., K,
\end{equation}
and 
\begin{equation}\label{eq:irf_dif4}
P(Y_{ij} = 1\vert \theta_i = \theta, x_i =0) = \frac{\exp(a_j \theta + d_j   )}{1+\exp(a_j \theta + d_j ) }.
\end{equation}
The structural model becomes 
$\theta_i \vert x_i=k \sim N(\beta_k,   \sigma_k^2), k = 1, ..., K,$
and $\theta_i \vert x_i=0 \sim N(0, 1)$. Under this model, 
an item $j$ is DIF-free if $\gamma_{jk} = 0$ for all $k$. The location indeterminacy under this model leads to the following ML1 condition for identifying the true model parameters $\Xi^* = \{\beta^*_k, (\sigma_k^*)^2, a_j^*, d_j^*,   \gamma_{jk}^*, k = 1, ..., K, j = 1, ..., J,\}$:
$$\sum_{j=1}^{J}  \vert \gamma_{jk}^* \vert < \sum_{j=1}^{J}  \vert \gamma_{jk}^* -a_j^* c_k\vert,$$
for $c_k\neq 0$, $k = 1, ..., K$. 

We note that this ML1 condition for the multi-group setting allows the majority of the items to be DIF items as long as the vector 
$(\gamma_{1k}^*, ..., 
\gamma_{Jk}^*)^\top$ is sufficiently sparse for each focal group. Similar to the discussion in Section \ref{subsec:ml1}, in the special case of the one-parameter logistic model, the ML1 condition is guaranteed to hold if $\sum_{j=1}^J I(\gamma_{jk}^* = 0) > J/2$, for all $k$. Note that the set of items satisfying $\gamma_{jk}^* = 0$ can vary across focal groups.

\paragraph{Continuous covariates.} In some applications, DIF might be caused by continuous covariates, such as age. Suppose that we have $K$ continuous covariates $\mathbf x_i =  (x_{i1}, ..., x_{iK})^\top$, rather than discrete groups. Then we may consider the following measurement model

\begin{equation}\label{eq:irf_dif5}
P(Y_{ij} = 1\vert \theta_i = \theta, \mathbf x_i) = \frac{\exp(a_j \theta + d_j +\boldsymbol\gamma_j^\top \mathbf x_{i} )}{1+\exp(a_j \theta + d_j +\boldsymbol\gamma_j^\top \mathbf x_{i})},
\end{equation}
where $\boldsymbol\gamma_j = (\gamma_{j1}, ..., \gamma_{jK})^\top$ be the corresponding DIF parameters. We may assume the structural model takes a homoscedastic latent regression form 
$\theta\vert \mathbf x_i \sim N(\boldsymbol\beta \mathbf x_i,1)$, where the variance is fixed to 1 to avoid scale indeterminacy\footnote{We note that the
homoscedastic assumption is commonly adopted in structural equation models. It is possible to extend the proposed method to a heteroscedastic structural model.}. 
Under this MIMIC model, an item $j$ is DIF-free if $\gamma_{jk} = 0$ for all $k$. 
The location indeterminacy under this model leads to the following ML1 condition for identifying the true model parameters $\Xi^* = \{\beta^*_k,  a_j^*, d_j^*,   \gamma_{jk}^*, k = 1, ..., K, j = 1, ..., J,\}$:
$$\sum_{j=1}^{J}  \vert \gamma_{jk}^* \vert < \sum_{j=1}^{J}  \vert \gamma_{jk}^* -a_j^* c_k\vert,$$
for $c_k\neq 0$, $k = 1, ..., K$. 

We note that this ML1 condition is similar to that under the multi-group setting. This is because the multi-group setting can be written in a very similar form as the current MIMIC model (by representing the groups using a covariate vector with dummy variables), except that the structural model under the multi-group setting allows heteroscedasticity. We also note that the current model assumes that a DIF effect is a linear combination of the covariates, which may seem inflexible, especially when comparing with the tree-based methods \citep{strobl2015rasch,tutz2016item,bollmann2018item}. However, we note that one can always move beyond the linearity by including transformations of the raw covariates (e.g., using spline basis) into the covariate vector and increasing the dimension of the DIF parameter vector $\boldsymbol\gamma_j$ simultaneously.

\paragraph{Ordinal response data.} Finally, we note that the proposed method can be extended to 
IRT models for other types of response data. To elaborate, we consider the generalized partial credit model (GPCM) \citep{muraki1992generalized} for ordinal response data as an example. 
For simplicity, we focus on the two-group setting (i.e., $x_i \in \{0,1\}$) and uniform DIF. Let $\{0, 1, ..., m_j\}$ be the ordered categories of item $j$. Then the measurement model becomes 
$$\frac{P(Y_{ij} = k\vert \theta_i = \theta, x_i)}{P(Y_{ij} = k-1\vert \theta_i = \theta, x_i)} = {\exp(a_j\theta + d_{jk} + \gamma_{jk}x_i)}, k = 1, ...,m_j,$$
where the DIF parameters $\gamma_{jk}$ depend on both the item and the category. 
We keep the structural model the same as in Section~\ref{subsec:str}. 
Under this model, 
an item $j$ is DIF-free if $\gamma_{jk} = 0$ for all $k$.
The location indeterminacy under this model leads to the following ML1 condition for identifying the true model parameters $\Xi^* = \{\beta^*, (\sigma^*)^2, a_j^*, d_j^*,   \gamma_{jk}^*, k = 1, ..., m_j, j = 1, ..., J\}$:
$$\sum_{j=1}^{J}\sum_{k=1}^{m_j} \vert \gamma_{jk}^* \vert < \sum_{j=1}^{J}\sum_{k=1}^{m_j} \vert \gamma_{jk}^* -a_j^* c\vert,$$
for all $c\neq 0$.

 }


\section{Simulation Study}\label{sec:sim}


This section conducts simulation studies to evaluate the performance of the proposed method and compare it with the  likelihood ratio test (LRT) method \citep{thissen1988use} and the LASSO method \citep{bauer2020simplifying}. Note that the LRT method requires a known anchor item set. Correctly specified anchor item sets with different sizes will be considered when applying the LRT method.

In the simulation, we set the number of items $J = 25$, and consider two settings for the sample sizes, $N = 500$, and 1000. The parameters of the true model are set as follows. First, the discrimination parameters are set between 1 and 2, {and we consider two sets of easiness parameters with one small $d_j$ set between $-1$ and 1 and another large $d_j$ set between $-2$ and 2, respectively.} Their true values are given in Table~\ref{tab: sim-parameters}. 
The observations are split into groups of equal sizes, indicated by $x_i = 0$, and 1.  The parameter $\beta$ in the structural model is set to 0.5 {and the parameter $\sigma$ is set to 0.5}, so that the latent trait distribution is standard normal {$N(0, 1)$}  and {$N(0.5, 0.5^2)$} for the reference and focal groups, respectively. {We consider six settings for the DIF parameters, three settings with DIF item proportions from high to low at smaller absolute DIF parameter values, and the other three with DIF item proportions from high to low at larger absolute  DIF parameter values. Specifically, at smaller and larger absolute DIF parameter values, the three settings contain 5, 10 and 14 DIF items out of 25 items for low, medium and high DIF proportions, respectively.} Their true values are given in Table~\ref{tab: sim-parameters}. For all sets of the DIF parameters, the ML1 condition is satisfied. The combinations of settings for the sample sizes and DIF parameters lead to 24 settings in total. For each setting, 100 independent datasets are generated.

We first evaluate the accuracy of the proposed estimator $\hat \Xi$ given by Algorithm~\ref{alg:ML1}. Table~\ref{tab: sim-mse} shows the mean squared errors (MSE) for $\beta$ {and $\sigma$} and the average MSEs for $a_j$s, $d_j$s, and $\gamma_j$s that are obtained by averaging the corresponding MSEs over the $J$ items. As we can see, these MSEs and average MSEs are small in magnitude and decrease as the sample size of individuals $N$ increases under each setting. This observation aligns with our consistency result in Theorem~\ref{thm:consistency}.

We then compare the proposed method and the LRT method in terms of their performances on statistical inference. Specifically, we focus on whether FDR can be controlled when applying the B-H procedure to the p-values obtained from the two methods. The comparison results are given in Table~\ref{tab: sim-FDR}. As we can see, FDR is controlled to be below the targeted level for   the proposed method and the LRT method with 1, 5, and 10 anchor items under all settings. 

When anchor items are known, the standard error can be computed for each estimated $\gamma_j$, and thus the corresponding Wald interval can be constructed. We compare the coverage rates of the confidence intervals given by Algorithm~\ref{alg:bootstrap} and the Wald intervals that are based on five anchor items. The results are shown in Figure~\ref{fig:coverage}. We see that the coverage rates from both methods are comparable across all settings and are close to the 95\% targeted level. Note that these coverage rates are calculated based on only 100 replicated datasets, which may be slightly affected by the Monte Carlo errors.


Finally, we compare the detection power of different methods based on the  receiver operating characteristic (ROC) curves. For a given method, 
a ROC curve is constructed by plotting the true positive rate (TPR) against the false positive rate (FPR) at different threshold settings. 
More specifically, ROC curves are constructed for the LASSO methods by varying the corresponding tuning parameters $\lambda$ from $0.02$ to $0.2$. ROC curves are also constructed by the LRT method with 
1, 5, and 10 anchor items, respectively. Note that for the LRT method, the TPR and FPR are calculated based on the non-anchor items.  
For each method, an average ROC curve is obtained based on the 100 replications, for which 
the area under the ROC curve (AUC) is calculated. A larger AUC value indicates better detection power. 
The AUC values for different methods across our simulation settings are given in Table \ref{tab: sim-AUC1}. 
According to the AUC values, the proposed procedure, that is, the p-value based method from Algorithm~\ref{alg:bootstrap}, 
performs better than the rest. That is, without knowing any anchor items, the proposed procedure performs better than the LRT method that knows 1 or 5 anchor items, and has similar performance as the LRT method that knows 10 anchor items under some settings with large DIF or large sample size $N$. 
The superior performance of the proposed procedures is brought by the use of the ML1 condition, which identifies the model parameters using information from all the items.
Based on the AUC values, we also see that the LASSO procedure performs similarly to the proposed procedures under some of the large DIF settings, but is less accurate under the small DIF settings.

\begin{figure}[h!]
    \centering
    \includegraphics[scale = 0.55]{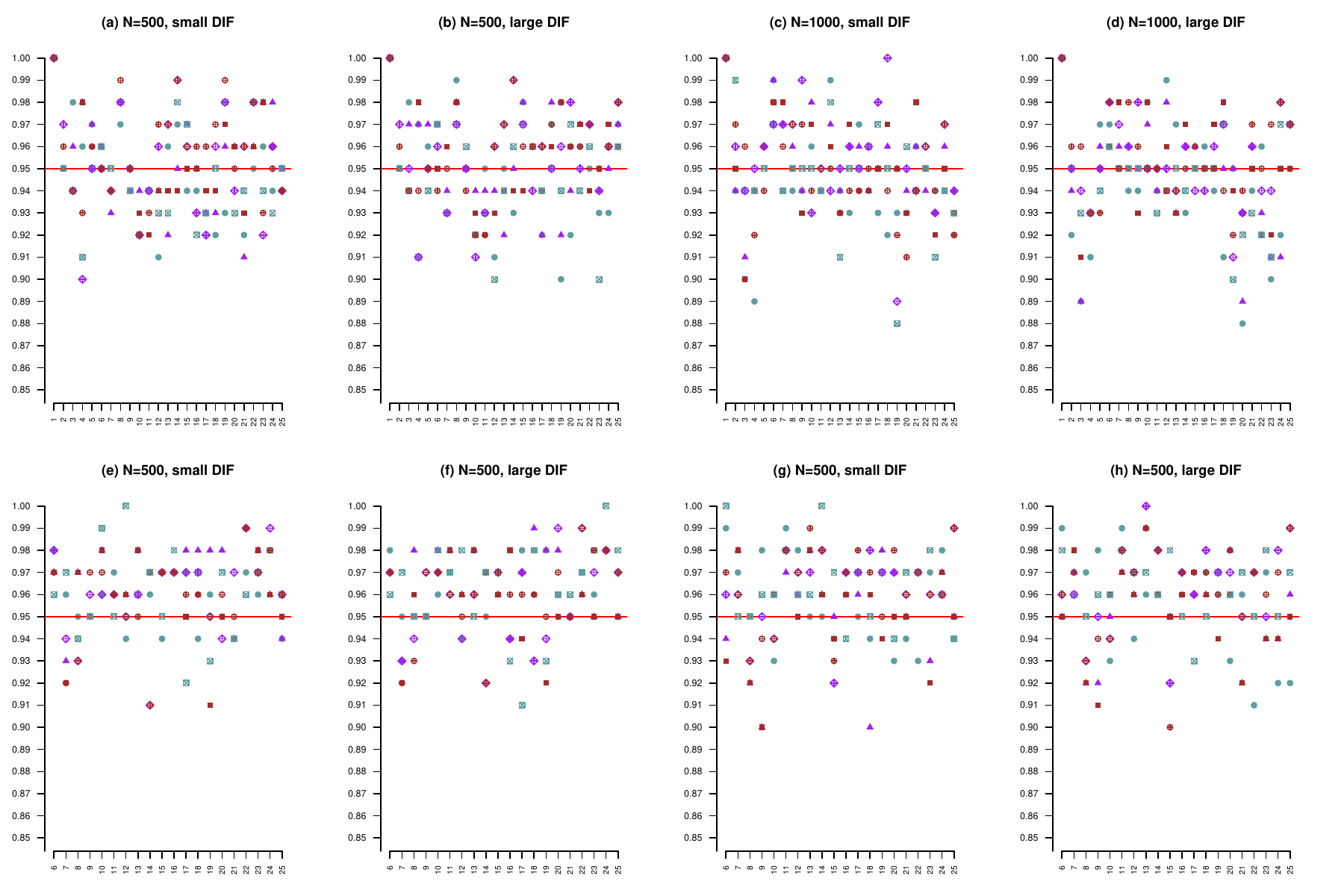}
    \caption{Scatter plots of the coverage rates of the 95\% confidence intervals for $\gamma_j^*$'s. x-axes and y-axes are labelled with item numbers and coverage rates, respectively. {Panels (a) - (d) correspond to our proposed method, and panels (e) - (h) correspond to the Wald intervals constructed with five anchor items. Blue solid circle (\protect\includegraphics[height=0.8em]{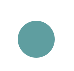}) corresponds to small $d_j$ with high proportion DIF items. Purple solid triangle (\protect\includegraphics[height=0.8em]{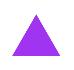}) corresponds to small $d_j$ with medium proportion DIF items. Red solid square (\protect\includegraphics[height=0.8em]{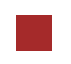}) corresponds to  small $d_j$ with low proportion DIF items. Blue square cross (\protect\includegraphics[height=1em]{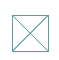}) corresponds to large $d_j$ with high proportion DIF items. Purple diamond plus (\protect\includegraphics[height=1em]{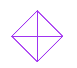}) corresponds to large $d_j$ with medium proportion DIF items. Red circle plus (\protect\includegraphics[height=1em]{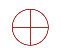}) corresponds to large $d_j$ with low proportion DIF items. }}
    \label{fig:coverage}
\end{figure}

\begin{table}[h!]
\small
\caption{Discrimination, easiness and DIF parameter values used in the simulation studies.}
\begin{center}
\begin{tabular}{c|c|c|c|c|c|c|c|c|c}
\hline 
 \multirow{2}{*}{Item number} & \multirow{2}{*}{$a_j$} & \multicolumn{2}{c|}{$d_j$} & \multicolumn{3}{c|}{$\gamma_j$ (Small DIF)} & \multicolumn{3}{c}{$\gamma_j$ (Large DIF)}
 \\
 \cline{3-10} 
  & & Small $d_j$ & Large $d_j$ & High & Medium & Low & High & Medium & Low \\
\hline
1 &  1.30 & 0.80 & 0.80 & 0.00 & 0.00 & 0.00 & 0.00 & 0.00& 0.00 \\ 
  2 &  1.40 & 0.20 & -0.40 & 0.00 & 0.00 & 0.00 & 0.00 & 0.00& 0.00 \\ 
  3 & 1.50 & -0.40 & -1.20 & 0.00 & 0.00 & 0.00 & 0.00 & 0.00& 0.00 \\ 
  4 &  1.70 & -1.00 & -2.00 & 0.00 & 0.00 & 0.00 & 0.00 & 0.00& 0.00 \\ 
  5 &  1.60 & 1.00 & 2.00 & 0.00 & 0.00 & 0.00 & 0.00 & 0.00& 0.00 \\ 
  6 &  1.30 & 0.80 & 0.80 & 0.00 & 0.00 & 0.00 & 0.00 & 0.00& 0.00 \\ 
  7 &  1.40 & 0.20 & -0.40 & 0.00 & 0.00 & 0.00 & 0.00 & 0.00& 0.00 \\ 
  8 &  1.50 & -0.40 & -1.20 & 0.00 & 0.00 & 0.00 & 0.00 & 0.00& 0.00 \\ 
  9 &  1.70 & -1.00 & -2.00 & 0.00 & 0.00 & 0.00 & 0.00 & 0.00& 0.00 \\ 
  10 & 1.60 & 1.00 & 2.00 & 0.00 & 0.00 & 0.00 & 0.00 & 0.00& 0.00 \\ 
  11 &  1.30 & 0.80 & 0.80 & 0.00 & 0.00 & 0.00 & 0.00 & 0.00 & 0.00 \\ 
  12 &  1.40 & 0.20 & -0.40 & -0.60 & 0.00 & 0.00 & -1.20 & 0.00 & 0.00 \\ 
  13 &  1.50 & -0.40 & -1.20 & 0.60 & 0.00 & 0.00 & 1.20 & 0.00 & 0.00  \\ 
  14 &  1.70 & -1.00 & -2.00 & -0.65 & 0.00 & 0.00 & -1.30 & 0.00 & 0.00 \\ 
  15 &  1.60 & 1.00 & 2.00 & 0.70 & 0.00 & 0.00 & 1.40 & 0.00 & 0.00 \\ 
  16 & 1.30 & 0.80 & 0.80 & -0.60 & -0.60 & 0.00 & -1.20 & -1.20 & 0.00 \\ 
  17 &  1.40 & 0.20 & -0.40 & 0.60 & 0.60  & 0.00 & 1.20 & 1.20 & 0.00 \\ 
  18 &  1.50 & -0.40 & -1.20 & -0.65 & -0.65 & 0.00 & -1.30 & -1.30 & 0.00 \\ 
  19 &  1.70 & -1.00 & -2.00 & 0.70 & 0.70 & 0.00 & 1.40 & 1.40 & 0.00 \\ 
  20 &  1.60 & 1.00 & 2.00 & 0.65 & 0.65 & 0.00 & 1.30 & 1.30 & 0.00 \\ 
  21 &  1.30 & 0.80 & 0.80 & -0.60 & -0.60 & -0.60 & -1.20 & -1.20 & -1.20 \\ 
  22 &  1.40 & 0.20 & -0.40 & 0.60 & 0.60 & 0.60 & 1.20 & 1.20 & 1.20  \\ 
  23 & 1.50 & -0.40 & -1.20 & -0.65 & -0.65 & -0.65 & -1.30 & -1.30 & -1.30 \\ 
  24 &  1.70 & -1.00 & -2.00 & 0.70 & 0.70 & 0.70 & 1.40 & 1.40 & 1.40 \\ 
  25 & 1.60 & 1.00 & 2.00 & 0.65 & 0.65 & 0.65 & 1.30 & 1.30 & 1.30\\ 
\hline
\end{tabular}
\label{tab: sim-parameters}
\end{center}
\end{table}

\begin{table}[h!]
\begin{center}
\caption{Average mean squared errors of the estimated parameters in the simulation studies. Mean squared errors are first evaluated by averaging out of 100 replications and then averaged across 25 items to obtain the average mean squared errors for $\bm{a}$, $\bm{d}$ and $\bm{\gamma}$. The mean squared errors for $\beta$ {and $\sigma$} are presented.}

\begin{tabular}{P{1.5cm}|P{1.5cm}|P{0.5cm}|P{1cm}P{1cm}P{1cm}|P{1cm}P{1cm}P{1cm} }
 \hline
\multicolumn{1}{c|}{}&\multicolumn{1}{c|}{}&\multicolumn{1}{c|}{}&\multicolumn{3}{|c}{Small DIF} &\multicolumn{3}{|c}{Large DIF}\\
 \hline
 &&& High & Medium & Low & High & Medium & Low\\
 \hline
\multirow{10}{*}{$N = 500$} &   \multirow{5}{*}{Small $d_j$} & $\bm{a}$ & 0.0482 & 0.0485  & 0.0485 & 0.0502 & 0.0490 & 0.0486 \\
 &   & $\bm{d}$ & 0.0316 & 0.0317 & 0.0318  & 0.0317 & 0.0317 & 0.0316 \\
 &    & $\bm{\gamma}$ & 0.0614 & 0.0612 & 0.0609 & 0.0670 & 0.0650 & 0.0623 \\
 &    & ${\beta}$ & 0.0010 & 0.0010  & 0.0010 & 0.0011 & 0.0011 & 0.0010\\
 &    & ${\sigma}$ & 0.0016 & 0.0017  & 0.0016 & 0.0017 & 0.0017 & 0.0018 \\
 \cline{2-9}
&  \multirow{5}{*}{Large $d_j$} & $\bm{a}$ & 0.0562 & 0.0552  & 0.0552 & 0.0589 & 0.0575 & 0.0560\\
 &    & $\bm{d}$ & 0.0467 & 0.0467 & 0.0470 & 0.0475 & 0.0476 & 0.0476\\
 &    & $\bm{\gamma}$ & 0.0873 & 0.0854 & 0.0834 & 0.1089 &0.1009 &  0.0903\\
 &    & ${\beta}$ & 0.0013 & 0.0012 & 0.0012 & 0.0011 & 0.0012 & 0.0013 \\
 &    & ${\sigma}$ & 0.0014 & 0.0014 & 0.0015 & 0.0015 & 0.0015 & 0.0015\\
  \cline{1-9}
 \multirow{10}{*}{$N = 1000$} &   \multirow{5}{*}{Small $d_j$} & $\bm{a}$ & 0.0227 & 0.0222 & 0.0222 & 0.0223 & 0.0221 & 0.0222\\
 &    & $\bm{d}$ & 0.0145 & 0.0145 & 0.0145 & 0.0145 & 0.0145 & 0.0145 \\
 &    & $\bm{\gamma}$ & 0.0291 & 0.0289  & 0.0287 &  0.0335 & 0.0320 & 0.0298 \\
 &    & ${\beta}$ & 0.0004 & 0.0004 & 0.0004 & 0.0005 & 0.0005 & 0.0004 \\
 &    & ${\sigma}$ &  0.0004 & 0.0005 & 0.0005 & 0.0005 & 0.0005 & 0.0005 \\
  \cline{2-9}
&  \multirow{5}{*}{Large $d_j$} & $\bm{a}$ & 0.0263 & 0.0261 & 0.0262 & 0.0270 & 0.0267 & 0.0264\\
 &    & $\bm{d}$ & 0.0223 & 0.0224 & 0.0225 & 0.0227 & 0.0224 & 0.0226\\
 &    & $\bm{\gamma}$ & 0.0412 & 0.0401 & 0.0392 & 0.0500 & 0.0461 & 0.0418 \\
 &    & ${\beta}$ & 0.0005 & 0.0005 & 0.0005 & 0.0006 & 0.0005 & 0.0005 \\
 &    & ${\sigma}$ & 0.0005 & 0.0005 & 0.0005 & 0.0006 & 0.0005 & 0.0005 \\
  \cline{1-9}
 \hline
\end{tabular}
\label{tab: sim-mse}
\end{center}
\end{table}

 \begin{table}
\begin{center}
\caption{Comparison of the FDR of the proposed p-value based method and the LRT method with 1, 5 and 10 anchor items respectively at the FDR control of 5\%.
The values are averaged out of 100 replications.}

\begin{tabular}{P{1.5cm}|P{1.4cm}|P{2cm}|P{1cm}P{1cm}P{1cm}|P{1cm}P{1cm}P{1cm} }
 \hline
\multicolumn{1}{c|}{}&\multicolumn{1}{c|}{}&\multicolumn{1}{c|}{}&\multicolumn{3}{c|}{Small DIF} &\multicolumn{3}{c}{Large DIF}\\
 \hline
 &&& High & Medium & Low & High & Medium & Low\\
 \hline
\multirow{8}{*}{$N = 500$} & \multirow{4}{*}{Small $d_j$} & proposed  & 0.0167 & 0.0255 & 0.0298  & 0.0192  & 0.0213  & 0.0319  \\
 &    & LRT 1 & 0.0089 & 0.0071 &  0.0137 &  0.0119 & 0.0148  & 0.0233  \\
 &    & LRT 5 & 0.0071 & 0.0181 & 0.0267  &  0.0122  & 0.0195  & 0.0394  \\
 &    & LRT 10 & 0.0033 &  0.0148 & 0.0283 & 0.0027 & 0.0154 & 0.0329\\
 \cline{2-9}
&  \multirow{4}{*}{Large $d_j$} & proposed  & 0.0240 & 0.0222  & 0.0323 & 0.0231 & 0.0249 & 0.0404\\
 &    & LRT 1 & 0.0164 & 0.0212 & 0.0267  & 0.0152  & 0.0216  & 0.0280  \\
 &    & LRT 5 & 0.0124 & 0.0221  & 0.0308 & 0.0128 & 0.0215 & 0.0246 \\
 &    & LRT 10 & 0.0031 & 0.0219  & 0.0237 & 0.0029 & 0.0159 & 0.0408\\
  \cline{1-9}
 \multirow{8}{*}{$N = 1000$} & \multirow{4}{*}{Small $d_j$} & proposed  & 0.0238 & 0.0277 & 0.0349 & 0.0229  & 0.0269  & 0.0425  \\
 &    & LRT 1 & 0.0087 & 0.0083 & 0.0152  & 0.0083  & 0.0131  & 0.0170  \\
 &    & LRT 5 & 0.0100 & 0.0217 & 0.0327 & 0.0087 & 0.0218 & 0.0341 \\
 &    & LRT 10 & 0.0021 &  0.0191 & 0.0389 & 0.0020 & 0.0162 & 0.0408\\
  \cline{2-9}
&  \multirow{4}{*}{Large $d_j$} & proposed  & 0.0217 & 0.0302 & 0.0390 & 0.0227  & 0.0333  & 0.0444  \\
 &    & LRT 1 & 0.0165 & 0.0166 & 0.0248  & 0.0172 & 0.0193  & 0.0237  \\
 &    & LRT 5 & 0.0114 & 0.0155 & 0.0249 & 0.0100 & 0.0162 & 0.0250 \\
 &    & LRT 10 & 0.0007 & 0.0062  & 0.0218 & 0.0013 & 0.0079 & 0.0260 \\
 \hline
\end{tabular}
\label{tab: sim-FDR}
\end{center}
\end{table}

\begin{table}[h!]
\begin{center}
\caption{Comparison of AUC of the proposed p-value based method, the LASSO method and the LRT method with 1, 5 and 10 anchor items respectively.}

\begin{tabular}{P{1.5cm}|P{1.4cm}|P{2cm}|P{1cm}P{1cm}P{1cm}|P{1cm}P{1cm}P{1cm} }
 \hline
\multicolumn{1}{c|}{}&\multicolumn{1}{c|}{}&\multicolumn{1}{c|}{}&\multicolumn{3}{c|}{Small DIF} &\multicolumn{3}{c}{Large DIF}\\
 \hline
 &&& High & Medium & Low & High & Medium & Low\\
 \hline
 \multirow{10}{*}{$N = 500$} & \multirow{5}{*}{Small $d_j$} & proposed  & 0.936 & 0.933 & 0.942  & 0.996 & 0.997  & 0.998   \\
&     & LASSO  & 0.802 & 0.805 & 0.789  & 0.992  & 0.991 & 0.987   \\
 &    & LRT 1 & 0.861 & 0.853 & 0.867  & 0.982 & 0.984 & 0.982   \\
 &    & LRT 5 & 0.915 & 0.917 & 0.920  & 0.992 & 0.991 &  0.988  \\
 &    & LRT 10 & 0.929 & 0.919 & 0.922  & 0.989 & 0.995 & 0.989   \\
 \cline{2-9}
&  \multirow{5}{*}{Large $d_j$} & proposed  & 0.910 & 0.915 & 0.917  & 0.986  & 0.988 & 0.990   \\
&     & LASSO  & 0.685 & 0.672 & 0.670 & 0.920 & 0.938 &  0.936  \\
 &    & LRT 1 & 0.823 & 0.800 & 0.826  & 0.966 & 0.966 & 0.969    \\
 &    & LRT 5 & 0.884 & 0.878 & 0.881  & 0.980 & 0.980 & 0.978   \\
 &    & LRT 10 & 0.897 & 0.875 & 0.884  & 0.983 & 0.975 & 0.977    \\
  \cline{1-9}
 \multirow{10}{*}{$N = 1000$} & \multirow{5}{*}{Small $d_j$} & proposed  & 0.984 & 0.986 & 0.987  & 1.000 & 1.000 & 1.000   \\
&     & LASSO  & 0.815 & 0.818 & 0.817 & 0.995  & 0.995   & 0.993   \\
 &    & LRT 1 & 0.965 & 0.968 & 0.960  & 0.997 &0.997  & 0.994   \\
 &    & LRT 5 & 0.979 & 0.975 & 0.976  & 0.990 & 0.990 & 0.990   \\
 &    & LRT 10 & 0.985 & 0.966 & 0.977  & 0.995 & 0.984 & 0.988   \\
 \cline{2-9}
&  \multirow{5}{*}{Large $d_j$} & proposed  & 0.964 & 0.964 & 0.965  & 0.997 & 0.998 & 0.998   \\
&     & LASSO  & 0.685 & 0.673 & 0.667  & 0.937 & 0.953 &  0.947  \\
 &    & LRT 1 & 0.944 & 0.942 & 0.941  & 0.989  & 0.995 & 0.992   \\
 &    & LRT 5 & 0.962 & 0.961 & 0.962  & 0.990  & 0.993 &  0.992  \\
 &    & LRT 10 & 0.972 & 0.953 & 0.962  & 1.000 & 0.998 & 0.992   \\
 \hline
\end{tabular}
\label{tab: sim-AUC1}
\end{center}
\end{table}

\section{Application to EPQ-R Data}\label{sec:real}

DIF methods have been commonly used for assessing the measurement invariance of personality tests (e.g., \citealp{escorial2007analysis}, \citealp{millsap2012statistical}, \citealp{thissen1986beyond}).  In this section, we apply the proposed method to  the Eysenck Personality Questionnaire-Revised (EPQ-R, \citealt{eysenck1985revised}), a personality test that has been intensively studied and received  applications worldwide \citep{fetvadjiev2015measures}. 
The EPQ-R has three scales that measure the Psychoticism (P), Neuroticism (N) and Extraversion (E) personality traits, respectively.  We analyze the long forms of the three personality scales that 
consist of 32, 24, and 23 items, respectively. Each item has binary responses of ``yes'' and ``no'' that are indicated by 1 and 0, respectively.  This analysis is based on data from an EPQ-R study collected from 1432 participants in the United Kingdom. Among these participants, 823 are females, and 609 are males. Females and males are indicated by $x_i = 0$ and 1, respectively. We study the DIF caused by gender. The three scales are analyzed separately using the proposed methods.

The results are shown through Tables~\ref{tab: setP}--\ref{tab: setN}, and Figure \ref{fig: real_CI}. Specifically, Tables ~\ref{tab: setP}--\ref{tab: setN} present the p-values from the proposed method for testing $\gamma_j = 0$ and the detection results for the P, E, N scales, respectively. 
For each table, the items are ordered by the p-values in increasing order. 
The items indicated by ``F'' are the ones detected by the B-H procedure with FDR level 0.05, and those indicated by {``L''} are the ones detected by 
{LASSO method whose tuning parameter $\lambda$ is chosen by BIC.}
The item IDs are consistent with those in Appendix 1 of \cite{eysenck1985revised}, where the item descriptions are given. The three panels of 
Figure \ref{fig: real_CI} further give the point estimate and confidence interval for each $\gamma_j$ parameter, for the three scales, respectively.
Under the current model parameterization, a positive DIF parameter means that a male participant is more likely to answer ``yes'' to the item than a female participant, given that they have the same personality trait level. We note that the absolute values of $\hat \gamma_j$ are all below 1, suggesting that there are no items with very large gender-related DIF effects. 


From Tables~\ref{tab: setP}--\ref{tab: setN}, we see that all three scales have some items whose p-values are close to zero, suggesting that gender DIF may exist across the three scales. The DIF items selected by the B-H procedure at the 5\% FDR level seem sensible. In what follows, we give some examples. \yc{For the P scale, the top four items are selected. These items are  ``14. Do you dislike people who don't know how to behave themselves?'', ``7. Would being in debt worry you?'',  ``34. Do you have enemies who want to harm you?'' and ``81. Do you generally ‘look before you leap’?'', with the DIF effect of item 7 being negative while those of the rest being positive. The discovery of items 14, 7 and 34 is consistent with the personality literature, where previous research has found that women are more gregarious and trusting  than men while men tend to be more risk-taking \citep{costa2001gender,feingold1994gender}. It is unclear from previous research why item 81 has a positive DIF effect. We conjecture that it is due to sociocultural influences. This result is consistent with that of  another P-scale item  ``2. Do you stop to think things over before doing anything?" whose statement is similar to item 81. Although not selected by the B-H procedure, the estimated DIF effect of this item is also positive, and its 95\% confidence interval does not include zero. 

For the E scale, eleven items are selected by the B-H procedure. Here, we discuss the top five items, including 
``63. Do you nearly always have a ‘ready answer’ when people talk to you?'', ``36. Do you have many friends?'',  ``90. Do you like plenty of bustle and excitement around you?'', ``6. Are you a talkative person?" and ``33. Do you prefer reading to meeting people?", where items 63 and 33  have positive DIF effects while the rest three have negative DIF effects. The discovery of these items is not surprising. 
The DIF effects of items 36, 90, 6 and 33 are consistent with previous observations that  women are more motivated to involve in social activities and tend to have more interconnected and affiliative social groups \citep{cross1997models}, which may be explained by the theory of self-construals \citep{markus1991culture}. The DIF effect of item 63 is consistent with the previous findings that men tend to score higher on assertiveness \citep{costa2001gender,feingold1994gender,weisberg2011gender}.

For the N scale, ten items are selected by the B-H procedure. Again, we discuss the top five items, including 
``8. Do you ever feel ‘just miserable’ for no reason?'', ``22. Are your feelings easily hurt?'', ``87. Are you easily hurt when people find fault with you or the work you do?'', ``84. Do you often feel lonely?" and ``70. Do you often feel life is very dull?", where items 8, 22 and 87 have negative DIF effects and items 84 and 70 have positive DIF effects. The discovery of items 8, 22, and 87 is consistent with the fact that women tend to score higher in tender-mindedness \citep{costa2001gender,feingold1994gender}. The positive DIF effects of items 84 and 70 may again be explained by the theory of self-construals \citep{markus1991culture}.

From Tables~\ref{tab: setP}--\ref{tab: setN}, we see that the selection based on the B-H procedure with FDR level 0.05 and that based on the LASSO procedure are quite consistent but do not exactly match.  
For the P-scale, the two procedures agree on four DIF detections, while the LASSO procedure additionally identifies four DIF items. For the E scale, they agree on six DIF detections, while the B-H procedure additionally identifies five items and the Lasso procedure additionally identifies one. Finally, for the N scale, the number of common detections is eight. Besides that, there are two items uniquely identified by the B-H procedure and four items uniquely identified by the Lasso procedure. 
Since the two procedures have different objectives (controlling FDR versus consistent model selection), it is not surprising that their results are not exactly the same. A consensus between the two methods suggests strong evidence, and thus, these common detections should draw our attention and be investigated first. For example, the content of the DIF items may be reviewed by experts, and new data may be collected to test these DIF effects through a confirmatory analysis. When there are enough resources, the items identified by one of the methods should also be investigated.}

\begin{figure}[h!]
    \centering
    \includegraphics[scale = 0.6]{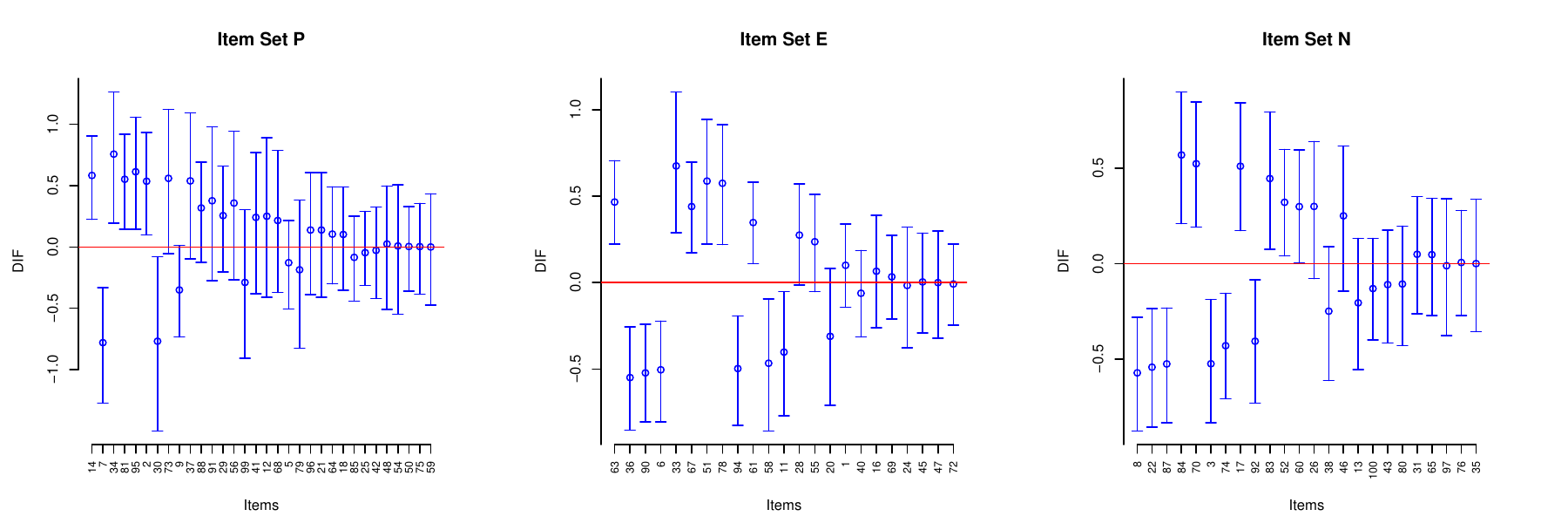}
    \caption{Plots of 95\% confidence intervals for the DIF parameters $\gamma_j^{*'}s$ on scale P, N, and E data sets. The red horizontal lines denote $\gamma=0$. Items are arranged according to the increasing p-values.}
    \label{fig: real_CI}
\end{figure}

\begin{table}[h!]
\small
\centering
\caption{P-values for testing $\gamma_j^*=0$ for items in P scale. Note that the items are ordered in increasing p-values. Items selected by the B-H procedure with FDR control at 5\% and {the LASSO method are identified using ``F'' and ``L''}, respectively, following the item numbers.}
\begin{tabular}{ c|c|c|c|c|c|c|c|c} 
\hline
Item & 14 F\textcolor{black}{L}  &7 F\textcolor{black}{L} &34 F\textcolor{black}{L}& 81 F\textcolor{black}{L}& 95 \textcolor{black}{L}&  2 \textcolor{black}{L}& 30  & 73 \\
\hline
p-value &0.0014& 0.0015& 0.0057& 0.0061& 0.0104& 0.0140& 0.0364 &0.0619\\
\hline
Item &9\textcolor{black}{L} &37 \textcolor{black}{L}&88 &91 &29& 56& 99 &41\\
\hline
p-value &0.0625& 0.0681 &0.1235 &0.2217 &0.2304 &0.2442 & 0.3389 &0.3780\\
\hline
Item &12 &68 & 5 &79 &96 &21 &64 &18\\
\hline
p-value & 0.4389 &0.4557& 0.4567& 0.5187 &0.5515& 0.5529 &0.5819 &0.5888\\
\hline
Item & 85& 25& 42& 48& 54& 50& 75& 59\\
\hline
p-value&0.6080 &0.7527 &0.8441 &0.8787&
0.9447 &0.9528 &0.9559 &0.9616\\
\hline
\end{tabular}
\label{tab: setP}
\end{table}

\begin{table} 
\small
\centering
\caption{P-values for testing $\gamma_j^*=0$ for items in E scale. Note that the items are ordered in increasing p-values. Items selected by the B-H procedure with FDR control at 5\% and \textcolor{black}{the LASSO method are identified using  ``F'' and ``L''}, respectively, following the item numbers.}
\begin{tabular}{ c|c|c|c|c|c|c|c|c} 
\hline
Item & 63 F\textcolor{black}{L}&36 F&90 F& 6 F&33 F\textcolor{black}{L}&67 F\textcolor{black}{L}&51 F\textcolor{black}{L}&78 F\textcolor{black}{L}\\
\hline
p-value &0.0000 &0.0004 &0.0006 &0.0011 &0.0013& 0.0013 &0.0016 &0.0019\\
\hline
Item & 94 F&61 F\textcolor{black}{L}&58 F&11 & 28\textcolor{black}{L}& 55& 20&  1\\
\hline
p-value &0.0031 &0.0051 &0.0199 &0.0310 &0.0644& 0.0958&0.1278& 0.4073\\
\hline
Item &40 &16 &69 &24& 45& 47& 72\\
\hline
p-value & 0.6185 &0.6439 &0.7819 &0.8371 &0.9291& 0.9364 &0.9391\\
\hline
\end{tabular}
\label{tab: setE}
\end{table}

\begin{table}[h!]
\small
\centering
\caption{P-values for testing $\gamma_j^*=0$ for items in N scale. Note that the items are ordered in increasing p-values. Items selected by the B-H procedure with FDR control at 5\% and {the LASSO method are identified using  ``F'' and ``L''}, respectively, following the item numbers.}
\begin{tabular}{ c|c|c|c|c|c|c|c|c} 
\hline
Item &8  F\textcolor{black}{L}&22  F\textcolor{black}{L}&87  F\textcolor{black}{L}&84  F\textcolor{black}{L}&70   F\textcolor{black}{L}&3  F&74  F\textcolor{black}{L}&17 F\textcolor{black}{L}\\
\hline
p-value &0.0004 &0.0006 &0.0007 &0.0014 &0.0016& 0.0026 &0.0026 &0.0037\\
\hline
Item &92  F&83  F\textcolor{black}{L}&52 \textcolor{black}{L} &60 \textcolor{black}{L} &26 \textcolor{black}{L} &38  &46 \textcolor{black}{L} &13 \\
\hline
p-value &0.0130 &0.0152 &0.0264 &0.0487 &0.0994& 0.1553&0.1856 &0.2337\\
\hline
Item & 100  &43  &80  &31  &65  &97  &76 & 35\\
\hline
p-value & 0.3365 &0.4417 &0.4694 &0.7116 &0.7376& 0.9220 &0.9531 &0.9550
\\
\hline
\end{tabular}
\label{tab: setN}
\end{table}

\section{Discussion}\label{sec:dis}

This paper proposes a new method for DIF analysis under a MIMIC model framework. It can accurately estimate the DIF effects of individual items without requiring prior knowledge about an anchor item set and can also provide valid p-values. 
The p-values can be used for the detection of DIF items and controlling the uncertainty in the decisions.  
According to our simulation results,  \yc{the proposed p-value-based procedure
has comparable performance in terms of classifying DIF and non-DIF items}, comparing with the LASSO method of \cite{belzak2020improving}. 
In addition, the p-value-based methods accurately control the item-specific type-I errors and the FDR. 
Finally, the proposed method is applied to the three scales of the Eysenck Personality Questionnaire-Revised to study gender-related DIF.
For each of the three long forms of the P, N, and E scales, around 10 items are detected by the proposed procedures as potential DIF items. The psychological mechanism of these DIF effects is worth further investigation. \yc{While the paper focuses on the two-group setting and uniform DIF, extensions to more complex settings have been discussed in Section~\ref{sec:related}, including non-uniform DIF, multi-group, and continuous covariate, and ordinal response settings.}
An R package has been developed for the proposed procedures that will be published online upon the acceptance of this paper. 

The proposed method has several advantages over the LASSO method. First, the proposed method does not require a tuning parameter to estimate the model parameters, while the LASSO method involves choosing the tuning parameter for the regularization term. Thus, the proposed method is more straightforward to use for practitioners.
Second, we do not need to solve optimization problems that involve maximizing a regularized likelihood function under different tuning parameter choices.
Therefore, the proposed method is computationally less intensive 
since the optimization involving a regularized likelihood function is non-trivial due to both the integral with respect to the latent variables and the non-smooth penalty term.
Finally, the proposed method provides valid statistical inference, which is more difficult for the LASSO method due to the uncertainty associated with the model selection step. With the obtained p-values, the proposed approach can detect the DIF items with controlled type-I error or FDR.  

The current work has some limitations, which offer opportunities for future research.  \yc{First, we note that the proposed method relies heavily on the ML1 condition, which tends to hold when the proportion of DIF-free items is high. While it may be sensible to make this assumption in many applications, there may also be applications where the proportion of DIF items is high, and thus, the ML1 condition fails to hold. Methods remain to be developed under such settings. 
One possible idea is to replace the $L_1$ norm in the ML1 condition with an $L_p$ norm for some $p \in (0,1)$. The $L_p$ norm better approximates the $L_0$ norm; thus, the corresponding condition is more likely to hold under a less sparse setting. However, the computation becomes more challenging when using the $L_p$ norm, as the transformation in Step 2 of Algorithm~\ref{alg:ML1} is no longer a convex optimization problem.} Second, 
as is true for
all simulation studies, we cannot examine all possible
conditions that might occur in applied settings. Additional simulation studies will be conducted in future research to understand the performance of the proposed method better. In particular, sample sizes, item sizes, group sizes and distribution of the DIF items can be varied and tested. \yc{Third, 
although the extensions to several more complex settings have been discussed in Section~\ref{sec:related}, these procedures remain to be implemented and assessed by simulation studies.}
Finally, the current work focuses on the type-I error and FDR as error metrics that concern falsely detecting non-DIF items as DIF items. In many applications of measurement invariance, it may also be of 
interest to consider an error metric that concerns the false detection of DIF items as DIF-free. 
Suitable error metrics, as well as methods for controlling such error metrics, remain to be proposed.

Although we focus on the DIF detection problem, the proposed method is also closely related to the problem of linking multiple groups' test results in the violation of measurement invariance \citep{asparouhov2014multiple,haberman2009linking,robitzsch2020lp}. \cite{robitzsch2020lp} proposed a linking approach based on an $L_p$ loss function, which is similar in spirit to the proposed method but focuses on linking multiple groups rather than DIF detection. We believe the proposed method can easily adapt to the linking problem to provide consistent parameter estimation and valid statistical inference. This problem is left for future investigation. 

\clearpage
\appendix

\noindent
{\Large Appendix}

\numberwithin{equation}{section}

This appendix contains additional proofs of all the proposition and theorems in Section \ref{app:proof} and discusses asymptotic distribution of $\tilde \Xi$ and the implementation details of Algorithms \ref{alg:ML1} and \ref{alg:bootstrap} in Section \ref{app:inference}. 

\section{Proofs of Propositions and Theorems}\label{app:proof}
 
 \begin{proof}[Proof of Proposition~\ref{prop:ML1}]
 Note $h$ is differentiable for all $c\neq 0$ with,  
 \begin{align*}
   \triangledown h(c)=\sum_{j=1}^J \vert a_j\vert \cdot sign(a_j^* c - \gamma_j^*),\quad \quad c\neq 0.
 \end{align*}
 Further note that $sign(a_j^* c - \gamma_j^*)=0$ when $c=\gamma_j^*/a_j^*.$ Hence we have
 \begin{align*}
&sign(a_j^* c - \gamma_j^*) >0 \quad \text{whenever}\quad c>\frac{\gamma_j^*}{a_j^*},\numberthis\label{eq: pos-sign-direction}\\
&sign(a_j^* c - \gamma_j^*)  < 0 \quad \text{whenever}\quad c<\frac{\gamma_j^*}{a_j^*}.\numberthis\label{eq: negative-sign-direction}
 \end{align*}
Consider the right derivative (positive directional derivative) of $h$ at $0$ from $+1$ direction, 
\begin{align*}
   \partial h^+(0) := \lim_{c \downarrow 0} \frac{h(c)-h(0)}{c}.
\end{align*} 
By the definition of right derivative of $h$ at $0$, \eqref{eq: pos-sign-direction} and \eqref{eq: negative-sign-direction}, we can rewrite $\partial h^+(0)$ equivalently as follows,
\begin{align*}
  \partial h^+(0)=   \sum_{j=1}^J |a_j^*| \left(- I\big(\frac{\gamma_j^*}{a_j^*} > 0\big) + I\big(\frac{\gamma_j^*}{a_j^*} \leq 0\big)\right).\numberthis\label{eq: pos-direvative}
\end{align*}
Similarly, define the left derivative (negative directional derivative) of $h$ at $0$ from $-1$ direction, $$\partial h^-(0):= \lim_{c \uparrow 0} \frac{h(c)-h(0)}{c}.$$
By the definition of left derivative $\partial h^-(0)$, \eqref{eq: pos-sign-direction} and \eqref{eq: negative-sign-direction}, we can rewrite $\partial h^-(0)$ equivalently as follows,
\begin{align*}
  \partial h^-(0)=  \sum_{j=1}^J |a_j^*| \left(- I\big(\frac{\gamma_j^*}{a_j^*} \geq 0\big) + I\big(\frac{\gamma_j^*}{a_j^*} < 0\big)\right).\numberthis\label{eq: neg-direvative}
\end{align*}
Since $h$ is convex, we must have $\argmin_{c} h(c) = 0$ if and only if $\partial h^+(0) >0$ and $\partial h^-(0)<0$ \citep{boyd2004convex, shor2012minimization}. From \eqref{eq: pos-direvative}, \eqref{eq: neg-direvative} and the fact that ML1 Condition \eqref{eq:ML1} is equivalent to $\argmin_c h(c) = 0$, the result of the proposition follows directly.
 \end{proof}

 \begin{proof}[Proof of Corollary~\ref{coro:1}]
  \textcolor{black}{By the definition of $\rho^*$, Condition~\eqref{eq:suff1} is equivalent to
     \begin{equation}
        \min_j\{|a_j^*|\} \sum_{j=1}^J I(\gamma_j^*/a_j^* \leq  0) > \max_j\{|a_j^*|\} \sum_{j=1}^J I(\gamma_j^*/a_j^* >  0). \nonumber
     \end{equation}
    For the left-hand side and right-hand side of the above inequality, we have
     \begin{eqnarray}
          \min_j\{|a_j^*|\} \sum_{j=1}^J I(\gamma_j^*/a_j^* \leq  0) &<& \sum_{j=1}^J |a_j^*| I(\gamma_j^*/a_j^* \leq  0); \nonumber \\
          \max_j\{|a_j^*|\} \sum_{j=1}^J I(\gamma_j^*/a_j^* >  0) &>& \sum_{j=1}^J|a_j^*| I(\gamma_j^*/a_j^* >  0). \nonumber
     \end{eqnarray}
     Therefore, Condition~\eqref{eq:suff1} implies
     \begin{equation}
         \sum_{j=1}^J |a_j^*| \left(I(\gamma_j^*/a_j^* \leq  0) - I(\gamma_j^*/a_j^* >  0)\right) > 0, \nonumber
     \end{equation}
which is~\eqref{eq:2prop1} in Proposition~\ref{prop:ML1}. Similarly, we have condition~\eqref{eq:suff2} implies
     \begin{equation}
          \sum_{j=1}^J |a_j^*| \left(I(\gamma_j^*/a_j^* <  0) - I(\gamma_j^*/a_j^* \geq  0)\right) > 0. \nonumber
     \end{equation}
which is~\eqref{eq:1prop1} in Proposition~\ref{prop:ML1}. Hence, if Conditions~\eqref{eq:suff1} and~\eqref{eq:suff2} are satisfied, we have Condition \eqref{eq:ML1} holds by Proposition~\ref{prop:ML1}.}
     
 \end{proof}

\begin{proof}[Proof of Theorem~\ref{thm:consistency}]
Since MIMIC model with constraint $\gamma_1^\dagger=0$ is identifiable, by classical asymptotic theory for MLE \citep{van2000asymptotic}, we have $\tilde \Xi$ converges in probability to $\Xi^\dagger.$ That is, as $N\to \infty$, for any $\epsilon>0$, we must have with probability tending to 1 that
$\vert \tilde \beta - \beta^\dagger\vert\leq \epsilon$, 
\textcolor{black}{$\vert \tilde \sigma^2 - (\sigma^2)^\dagger\vert\leq \epsilon$,}
$\vert \tilde \gamma_j - \gamma_j^\dagger\vert \leq \epsilon, \vert \tilde a_j - a_j^\dagger\vert \leq \epsilon$ and $\vert \tilde d_j - d_j^\dagger\vert \leq \epsilon$, for any $j=1,...,J$. Denote $f(c)=\sum_{j=1}^J \vert \gamma_j^\dagger - c a_j^\dagger\vert$ as a function of $c.$ Similarly, denote $f_N(c)=\sum_{j=1}^J \vert \tilde\gamma_j - c \tilde a_j\vert$. Let $c^\dagger = \arg\min_{c} f(c)$ and $\hat c = \arg\min_{c} f_N(c)$, respectively. We seek to establish that $\hat c$ will converge in probability to $c^\dagger$ as $N\to \infty.$ 
First note that by regularity conditions, there exists $C_1<\infty$ such that $J, \vert \gamma_j^\dagger\vert, \vert a_j^\dagger\vert \leq C_1.$ Then, there must exist $C_2<\infty$ such that $\vert c^\dagger \vert \leq C_2.$ Furthermore, note $f_N$ is clearly continuous and convex in $c$, so consistency will follow if $f_N$ can be shown to converge point-wise to $f$ that is uniquely minimized at the true value $c^\dagger$ (typically uniform convergence is needed, but point-wise convergence of convex functions implies their uniform convergence on compact subsets). 
Following the model identifiability and the ML1 condition (\ref{eq:ML1}), $c^\dagger$ is unique. To see this, suppose for contradiction that there exist $c_1$ and $c_2$ such that $c_1\neq c_2$ and $c_1=\arg\min_{c} f(c)$ and $c_2=\arg\min_{c} f(c).$ 
First note that $a_j^\dagger=a_j^*$ for all $j=1,...,J.$
Then by model identifiability, there exists $c_3$ such that $\gamma_j^\dagger=\gamma_j^*+c_3 a_j^*.$
So we have $$c_1=\arg\min_{c} \sum_{j=1}^J \vert \gamma_j^* + (c_3-c) a_j^*\vert$$ and $$c_2=\arg\min_{c} \sum_{j=1}^J \vert \gamma_j^* + (c_3-c) a_j^*\vert.$$
Hence, $\gamma^*=\gamma^\dagger + (c_3-c_1)a_j^*$ and $\gamma^*=\gamma^\dagger + (c_3-c_2)a_j^*$. If ML1 condition (\ref{eq:ML1}) holds, then $c_3=c_1$ and $c_3=c_2.$ This contradicts the assumption $c_1\neq c_2.$ Hence, $c^\dagger$ must be unique. 
For any $\vert c\vert \leq C_2,$\\
\begin{align*}
 & \vert f_N(c) - f(c)\vert\\
& = \Big\vert \sum_{j=1}^J \Big(\vert \tilde\gamma_j - c \tilde a_j\vert-\vert \gamma_j^\dagger - c a_j^\dagger\vert\Big) \Big\vert \\
&\leq \Big\vert \sum_{j=1}^J \Big(\vert (\tilde\gamma_j - c \tilde a_j)- ( \gamma_j^\dagger - c a_j^\dagger)\vert\Big) \Big\vert \\
&=\Big\vert \sum_{j=1}^J \Big(\vert (\tilde\gamma_j - \gamma_j^\dagger)+ c( a_j^\dagger- \tilde a_j)\vert\Big) \Big\vert \\
&\leq \sum_{j=1}^J \Big(\vert \tilde\gamma_j - \gamma_j^\dagger\vert+ \vert c \vert \cdot \vert a_j^\dagger- \tilde a_j\vert\Big) \\
&\leq J \epsilon + \vert c \vert\epsilon.\\
&\leq (C_1+C_2)\epsilon.\\
\end{align*}
Take $\epsilon_1=(C_1+C_2)\epsilon$, it follows that for any fixed $\vert c\vert \leq C_2$, $P\big(\vert f_N(c) - f(c)\vert \leq \epsilon_1\big) \to 1$ as $N\to \infty$. Moreover, following from the uniqueness of $c^\dagger$ and the continuity and the convexity of $f_N(\cdot)$ in $c$, we must have $\vert \hat c - c^\dagger \vert = o_P(1)$ as $N \to \infty.$ 

Note that $\hat \beta=\tilde\beta + \hat c$, \textcolor{black}{$\hat{\sigma}^2 = \tilde{\sigma}^2$}, $\hat \gamma_j = \tilde \gamma_j - \hat c \tilde a_j$, $\hat a_j=\tilde a_j$, $\hat d_j=\tilde d_j$ for all $j=1,...,J.$ From the model identifiability and the ML1 condition (\ref{eq:ML1}), we know that $\beta^*=\beta^\dagger + c^\dagger$, \textcolor{black}{$(\sigma^2)^* = (\sigma^2)^\dagger$},  $\gamma_j^* = \gamma_j^\dagger - c^\dagger a_j^\dagger$, $a_j^*=a_j^\dagger$, $d_j^*=d_j^\dagger$ for all $j=1,...,J.$ Since $\vert \hat c - c^\dagger \vert = o_P(1), \vert \tilde \beta - \beta^\dagger\vert=o_P(1)$, \textcolor{black}{$\vert \tilde{\sigma}^2 - (\sigma^2)^\dagger\vert=o_P(1)$,} $\vert \tilde \gamma_j - \gamma_j^\dagger\vert =o_P(1), \vert \tilde a_j - a_j^\dagger\vert =o_P(1), \vert \tilde d_j - d_j^\dagger\vert =o_P(1)$ as $N\to\infty$, it follows directly from the Slutsky's Theorem that $\vert \hat \beta - \beta^*\vert=o_P(1), \textcolor{black}{\vert \hat{\sigma}^2 - (\sigma^2)^*\vert=o_P(1),} \vert \hat \gamma_j - \gamma_j^*\vert =o_P(1), \vert \hat a_j - a_j^*\vert =o_P(1)$, $\vert \hat d_j - d_j^*\vert =o_P(1)$ as $N\to \infty$. 
\end{proof}

\section{Asymptotic Distribution of $\tilde \Xi$}\label{app:inference}
Since the model is identifiable with constraint $\gamma_1^{\dagger}=0$
and all the regularity conditions in Theorem 5.39 of \cite{van2000asymptotic} are satisfied, hence, by Theorem 5.39 in \cite{van2000asymptotic}, $\tilde \Xi \to $ N$(\Xi^\dagger, \Sigma^*)$ in distribution as $N\to \infty.$ In practice, we use the inverse of the observed Fisher information matrix, denoted by  $\hat \Sigma_N$, which is a consistent estimator of $\Sigma^*$, to draw Monte Carlo samples. Below, we give procedures to evaluate $\hat \Sigma_N$ from the marginal log-likelihood.

Following the notations in the main article, we first provide the complete data log-likelihood function,
\begin{align*}
l(\Xi; Y)=&\sum_{i=1}^{N} \Big[\log \Big\{\frac{1}{\sqrt{2\pi(\textcolor{black}{1_{\{x_i = 0\}} + \sigma^2 1_{\{x_i = 1\}}})}} \exp\left(\frac{-(\theta_i - \beta x_i)^2}{2(\textcolor{black}{1_{\{x_i = 0\}} + \sigma^2 1_{\{x_i = 1\}}})}\right)\Big\}\\
&+ \sum_{j=1}^{J}\big\{ y_{ij}(a_j\theta_i + d_j +\gamma_j x_i) - \log(1+\exp\{a_j\theta_i + d_j +\gamma_j x_i\})\big\}\Big].
\end{align*}
Since $\theta_i$ is considered as a random variable such that \textcolor{black}{$\theta_i\mid x_i \sim $ N$(\beta x_i, 1_{\{x_i = 0\}} + \sigma^2 1_{\{x_i = 1\}})$}, so we will work with the marginal log-likelihood function,
\begin{equation*}
mll(\Xi; Y) = \sum_{i=1}^N \log \Big\{\int \left(\prod_{j=1}^J \frac{\exp(y_{ij}(a_j\theta_i + d_j + \gamma_j x_i))}{1+\exp(a_j\theta_i + d_j + \gamma_j x_i)}\right)\frac{1}{\sqrt{2\pi}} \exp\left(\frac{-(\theta_i - \beta x_i)^2}{2(\textcolor{black}{1_{\{x_i = 0\}} + \sigma^2 1_{\{x_i = 1\}}})}\right)  d\theta_i\Big\}.
\end{equation*}

Note that the observed Fisher information matrix $I(\Xi)$ cannot be directly obtained from the $mll(\Xi; Y)$ due to the intractable integral. Instead,  we apply the Louis Identity \citep{louis1982finding} to evaluate the observed Fisher information matrix. Let $S(\Xi; Y)$ and $B(\Xi; Y)$ denote the gradient vector and the negative of the hessian matrix of the complete data log-likelihood function, respectively. Then by the Louis Identity,  $I(\Xi)$ can be expressed as
\begin{align*}
I(\Xi)= \mathbf{E}_{\theta}[B(\Xi; Y) \mid Y] - \mathbf{E}_{\theta}[S(\Xi; Y)S(\Xi; Y)^T \mid Y]+\mathbf{E}_{\theta}[S(\Xi; Y) \mid Y]\mathbf{E}_{\theta}[S(\Xi; Y) \mid Y]^T.
\end{align*}
Denote $p_{ij}=\exp\{y_{ij}(a_j\theta_i + d_j + \gamma_j x_i)\}/[1+\exp\{y_{ij}(a_j\theta_i + d_j + \gamma_j x_i)\}]$. Then, in particular, 
\begin{align*}
S(\Xi; Y)= &\frac{\partial l(\Xi; Y)}{\partial \Xi}\\
=&\Big\{\frac{\partial l(\Xi; Y)}{\partial \beta},\textcolor{black}{\frac{\partial l(\Xi; Y)}{\partial \sigma^2}},...,\frac{\partial l(\Xi; Y)}{\partial a_j},...,\frac{\partial l(\Xi; Y)}{\partial d_j},...,\frac{\partial l(\Xi; Y)}{\partial \gamma_j},...\Big\}\\
=&\Big\{\textcolor{black}{\frac{\sum_{i=1}^N x_i(\theta_i-\beta)}{\sigma^2}, \frac{\sum_{i=1}^N x_i(\theta_i-\beta)^2}{2\sigma^4} - \frac{\sum_{i=1}^N x_i}{2\sigma^2}, } ...,\\
&\sum_{i=1}^N\theta_i (y_{ij}-p_{ij}),...,\sum_{i=1}^{N}(p_{ij}-y_{ij}),..., \sum_{i=1}^{N} x_i (y_{ij}-p_{ij})\Big\}.
\end{align*}
Furthermore, note that $B(\Xi; Y) = -\partial^2 l(\Xi; Y)/\partial \Xi\partial\Xi^T$ is a \textcolor{black}{$(3J+2)$ by $(3J +2)$} matrix with the only non-zero entries,
\begin{align*}
&\textcolor{black}{\frac{\partial^2 l(\Xi; Y)}{\partial \beta^2}=-\frac{\sum_{i=1}^{N} x_i}{\sigma^2},}\\
&\textcolor{black}{\frac{\partial^2 l(\Xi; Y)}{\partial (\sigma^2)^2}= -\frac{\sum_{i=1}^N x_i(\theta_i-\beta)^2}{\sigma^6} + \frac{\sum_{i=1}^N x_i}{2\sigma^4},}\\
&\textcolor{black}{\frac{\partial^2 l(\Xi; Y)}{\partial \beta\partial \sigma^2 }=-\frac{\sum_{i=1}^{N} x_i(\theta_i - \beta)}{\sigma^4},}\\
&\frac{\partial^2 l(\Xi; Y)}{\partial a_j^2}=-\sum_{i=1}^{N} \theta_i^2p_{ij}(1-p_{ij}),\\
&\frac{\partial^2 l(\Xi; Y)}{\partial d_j^2}=-\sum_{i=1}^{N} p_{ij}(1-p_{ij}),\\
&\frac{\partial^2 l(\Xi; Y)}{\partial \gamma_j^2}=-\sum_{i=1}^{N} x_i^2 p_{ij}(1-p_{ij}),\\
&\frac{\partial^2 l(\Xi; Y)}{\partial a_j \partial d_j}=\sum_{i=1}^{N} \theta_i p_{ij} (1-p_{ij}),\\
&\frac{\partial^2 l(\Xi; Y)}{\partial a_j \partial \gamma_j}= -\sum_{i=1}^{N} \theta_i x_i p_{ij} (1-p_{ij}),\\
&\frac{\partial^2 l(\Xi; Y)}{\partial d_j \partial \gamma_j}= \sum_{i=1}^{N} x_i p_{ij} (1-p_{ij}).
\end{align*}
In practice, we can use Gaussian quadrature method to approximate the expectation of these terms so as to obtain $\hat I(\tilde \Xi)$. Then $\hat \Sigma_N$ can be evaluated with $\hat \Sigma_N = \hat I^{-1}(\tilde \Xi)$. This then enables Step 1 of Algorithm 1, where Monte Carlo samples of $\Xi^\dagger$ can be simulated from N$(\tilde \Xi, \hat \Sigma_N).$

\clearpage

\bibliography{bibliography}

\end{document}